\documentclass[10pt, doublecolumn, romanappendices]{IEEEtran}
\usepackage[draft, blank]{ieeefig}
\usepackage{float}
\floatstyle{plaintop}
\restylefloat{table}
\usepackage{todonotes}

\usepackage{textcmds}

\usepackage{mdframed}
\usepackage{subfigure}
\usepackage{cite}
\usepackage{stfloats}
\usepackage{relsize}
\usepackage{amsopn}
\usepackage{float}
\restylefloat{table}

\usepackage{amsxtra}
\usepackage{amsmath}
\usepackage{graphicx}
\usepackage{latexsym}
\usepackage{amsfonts}
\usepackage{amssymb}
\usepackage{mathrsfs}
\usepackage{bm}
\usepackage{yhmath}
\usepackage{verbatim}
\usepackage{float}
\usepackage{xr-hyper}
\usepackage{multirow}
\usepackage{subfloat}

\usepackage{epstopdf}

\DeclareMathAlphabet{\bi}{OML}{cmm}{b}{it}
\DeclareMathAlphabet{\bcal}{OMS}{cmsy}{b}{n}
\DeclareMathAlphabet{\brmn}{OT1}{cmr}{bx}{n}

\DeclareMathSymbol{\R}{\mathalpha}{AMSb}{"52}

\newcommand{\bgamma}{\boldsymbol{\gamma}}

\newcommand{\brho}{\boldsymbol{\rho}}

\newcommand{\bxi}{\boldsymbol{\xi}}

\newcommand{\Ac}{\mathcal{A}}

\newcommand{\Fc}{\mathcal{F}}

\newcommand{\Lc}{\mathcal{L}}

\newcommand{\Oc}{\mathcal{O}}
\newcommand{\Pc}{\mathcal{P}}

\newcommand{\Rc}{\mathcal{R}}

\newcommand{\Cb}{\mathbb{C}}

\newcommand{\Rb}{\mathbb{R}}

\newcommand{\Zb}{\mathbb{Z}}

\renewcommand{\o}{\omega}
\usepackage{amssymb}
\usepackage{amsfonts}

\graphicspath{{./images/}}

\DeclareMathAlphabet{\bi}{OML}{cmm}{b}{it}
\DeclareMathAlphabet{\bcal}{OMS}{cmsy}{b}{n}
\DeclareMathAlphabet{\brmn}{OT1}{cmr}{bx}{n}

\def \x{\mathbf{x}}

\def \y{\mathbf{y}}

\usepackage{algorithm}
\usepackage{algorithmic}

\usepackage{amsthm}

\newtheorem{theorem}{Theorem}

\newtheorem{lemma}[theorem]{Lemma}

\newtheorem{definition}{Definition}

\title{Performance Analysis of Convex LRMR based \\ Passive SAR Imaging}

\author{Eric Mason$^1$ and Birsen Yaz{\i}c{\i}$^{1*}$, \IEEEmembership{Senior Member, IEEE}
\thanks{$^1$ Yaz{\i}c{\i} and Mason are with the Department of Electrical, Computer and Systems Engineering,
Rensselaer Polytechnic Institute, 110 8th Street, Troy, NY 12180 USA E-mail: B.~Y: yazici@ecse.rpi.edu,
Phone: (518)-276 2905, Fax: (518)-276 6261.}
\thanks{$^*$ Corresponding author.}
 \thanks{This work was supported by the Air Force Office of
 Scientific Research (AFOSR) under the agreement FA9550-16-1-0234, and by industrial and government membership fees to the Center for Surveillance Research, a National Science Foundation Industry and University Cooperative Research Program (I/UCRC) under agreement IIP-1539961. }}

\begin{document}

\maketitle

\begin{abstract}

Passive synthetic aperture radar (SAR) uses existing signals of opportunity such as communication and broadcasting signals.
In our prior work, we have developed a low-rank matrix recovery (LRMR) method that can reconstruct scenes with extended and densely distributed point targets, overcoming shortcomings of conventional methods \cite{Mason15}.
The approach is based on correlating two sets of bistatic measurements, which results in a linear mapping of the  tensor product of the scene reflectivity with itself.
Recognizing this tensor product as a rank-one positive semi-definite (PSD) operator, we pose passive SAR image reconstruction as a LRMR problem with convex relaxation.

In this paper, we present a performance analysis of the convex LRMR-based passive SAR image reconstruction method.
We use the restricted isometry property (RIP) and show that exact reconstruction is guaranteed under the condition that the pixel spacing or resolution satisfies a certain lower bound.
We show that for sufficiently large center frequencies, our method provides superior resolution than that of Fourier based methods,
making it a super-resolution technique.
Additionally, we show that phaseless imaging is a special case of our passive SAR imaging method.
We present extensive numerical simulation to validate our analysis.

\end{abstract}

\section{Introduction}

A passive synthetic aperture radar (SAR) uses signals of opportunity such as communication or broadcasting signals to form an image. Passive radars offer several advantages including reduced cost, simplicity of hardware, increased stealth and efficient use of electromagnetic spectrum. With the proliferation of transmitters of opportunity, there is a growing interest in passive radar applications \cite{Yarman10,Wang11,Wacks14,Wacks14_2,Arroyo13,Hack2014_2,Palmer13,Maslikowski14,Wang10}.


Widely used existing methods for the passive SAR problem can be categorized into two classes: passive coherent localization (PCL) \cite{Baker05_1,Baker05_2,Wu2015,Hack2014_2,Kulpa12,colone2009} and time-difference-of-arrival (TDOA) based backprojection \cite{Yarman10,Yarman08,Wang11,Wacks14_2,Wacks14,Qu2017,Shi2016}.
These approaches rely on certain assumptions and imaging configurations that limit their applicability.
Specifically, PCL requires a direct line-of-sight to a transmitter of opportunity and knowledge of transmitter location, and TDOA-based backprojection relies on a point target assumption on the scene.
In \cite{Mason15,Mason2015b}, we presented a low-rank matrix recovery (LRMR)-based alternative method that overcomes these shortcomings.
This method can image scenes with extended or densely distributed point targets, and does not require direct-line-of-sight to the transmitter or knowledge of the transmitted waveform \cite{Mason15}.

We assume that two or more spatially separated moving receivers acquire scattered field data from a scene of interest illuminated by transmitters of opportunity. We next cross-correlate received signals from different receivers. 
The correlation process induces a linear mapping from the tensor product of the scene reflectivity with itself to cross-correlated measurements.
We linearize the problem by using the well-known \qq{lifting} approach, by noting that the tensor product of the scene reflectivity with itself is the kernel of a rank-one positive semi-definite (PSD) operator \cite{Demanet14,Candes13b,Bayram2017,Chai11}.
Thus, we approach image reconstruction as a rank minimization problem and pose it
using the well-known nuclear-norm as a convex relaxation \cite{Recht10,Fazel02}.

In this paper, we present performance analysis of convex LRMR-based passive SAR imaging method. Specifically, we show that under the condition that the distance between scatterers satisfies a certain lower bound, the recovery is exact, i.e.,
nuclear-norm minimization algorithm
 converges and recovers the location and reflectivities of scatterers exactly. The condition on the distance between the scatterers determines the highest achievable resolution for exact recovery. Given the bandwidth and operating frequencies of typical illuminators of opportunity, the achievable resolution is an order of magnitude higher than that of Fourier based techniques, making the convex LRMR-based passive SAR imaging a super-resolution technique.

Our analysis is based on the use of the restricted isometry property (RIP), one of several tools in LRMR theory used to establish exact reconstruction guarantees when solving
nuclear-norm minimization problems
\cite{Recht10,Oymak2011}.
This property quantifies how close the forward model is to an isometry in terms of the restricted isometry constant (RIC), when the forward map is restricted to a set of rank-r matrices. We show that under certain conditions the RIC associated with the passive SAR forward model is small enough to guarantee that the solution to the
nuclear-norm minimization problem
is exact.
We use the method of stationary phase in determining the RIC. 

 Nuclear-norm minimization problem can be solved in a variety of ways \cite{Recht10,Cai2010,Combettes2011,Combettes2005}. 
Here, we pose it as a trace-minimization, PSD constrained problem with an affine data-fidelity constraint,
and use Uzawa's method to solve the resulting saddle point problem\footnote{ This approach is different from the trace regularized least-squares formulation we used in \cite{Mason15}. The new formulation allows us to establish its equivalence to the nuclear norm-minimization problem that the RIP based recovery results apply.}.  
We then prove convergence to a unique minimizer, which recovers the exact unknown scene reflectivity when the RIP holds.  

Additionally, we show that phaseless imaging is a special case of our LRMR-based passive imaging method. 
Assuming the two receivers are co-located our passive SAR forward model reduces to an auto-correlation. 
This type of configuration is known as phaseless imaging \cite{Demanet14,Candes13b,Demanet13,Candes13a}.
Recovery results have been obtained for active phaseless imaging using a rectangular array in \cite{Chai11}.
Our analysis covers the phaseless imaging case and generalizes the result in \cite{Chai11} to arbitrary imaging geometries and passive imaging. It also provides a sufficient condition that proves it to be a super-resolution technique.
Additionally, we show that cross-correlation results in a smaller RIC than that of auto-correlated measurements, thus requiring lower center frequencies to obtain exact reconstruction.

The rest of this paper is organized as follows: in Section \ref{sec:PA_forward_model} we present the forward model and discretization of the problem. 
In Section \ref{sec:PA_LRMR_theory}, we formulate a set of optimization problems for passive SAR imaging , present an optimization algorithm and use the relevant LRMR theory used in our analysis.
Section \ref{sec:PA_analysis_gaurantees} presents our main result and a sufficient condition for exact recovery. Section \ref{sec:phaseless} provides an extension of our analysis to phaseless imaging. In Section \ref{sec:results_discuss} we discuss the implications of our analysis and the validity of our assumptions. Section \ref{sec:numerical_simulations} presents numerical simulations.  Section \ref{sec:conclusion} concludes the paper.

\begin{table}[!htbp]
\caption{Notation} 
  \vspace{.1cm}
\centering 
{\footnotesize
  \begin{tabular}{p{1.3in}p{1.5in}}
  	\hline \hline
  Symbol & Description \\ 
  \hline 
$\omega$ & Fast-time temporal frequency \\
$s$ & Slow-time \\
$\bgamma_{i}(s)$ & Trajectory for $i^{th}$ receiver \\
$\y$ & Transmitter location \\
$\hat{\y}$ & Unit vector in transmitter direction: $\y/|\y|$ \\
$\alpha_{\y}$ & Transmitter elevation angle \\
$\psi(\bi x)$ & Ground topography function \\
$\brmn x = [\bi x, \psi(\bi x)]$ & Location on earth's surface \\
$\tilde{\rho}(\bi x)$ & Surface reflectivity \\
$\otimes$ & Kronecker product symbol \\
$f_i(\omega,s)$ & Received signal $i$-th the receiver \\
$d_{ij}(\omega,s)$ & Correlated measurements for $i,j$-th receivers \\
$\Lc_i$ & Bistatic forward mapping for the $i$-th receiver \\ 
$\Fc_{ij}$ & Forward map for the $i,j$-th data \\ 
$A_{ij}$ & Amplitude function of the forward map \\ 
$\phi_{ij}$ & Phase function of the forward map \\ 
$\rho(\bi x, \bi x')$ & \textit{Kroneker scene} \\ 
$\Rc$ & Operator with kernel $\rho$ \\ 
$\mathbf{F}$ & Discrete forward model \\
$\mathbf{d}$ & Discretized and vectorized data \\
$\brho$ $(\bar{\brho})$ & Discrete (and vectorized) Kronecker scene \\
$L_{g_i}$ & Distance from scene center to $i$-th receiver ($|\bgamma_i(s)|$) \\
$L_x$ & Scene size \\
$B$ & Bandwidth \\
$\omega_c$ & Center frequency \\
$\Delta_x$ & Minimum pixel (target) spacing \\
$S$ &  Slow-time set \\
$\Omega$ & Baseband frequency set \\
\hline \hline
  \end{tabular}}
  \label{tab:notation}
\end{table} 

\section{Passive Imaging Forward Model}\label{sec:PA_forward_model}

We assume there are two airborne receivers flying over a scene of interest, illuminated by a transmitter of opportunity located at $\y \in \Rb^3$.
We denote the receiver trajectories by $\bgamma_{i}(s) \in \Rb^3, i = 1,2, s \in S := [s_a,s_b]$ where $s$ represents the slow-time variable.
An example of this configuration is displayed in Figure \ref{fig:passive_sar_fig_chap4}.

$\o \in [\o_c-B/2,\o_c+B/2]$ be the fast-time temporal frequency variable, where $B$ is the bandwidth and $\o_c$ is the center frequency. 
$\Omega = [ -B/2,B/2 ]$ is the baseband frequency set.
 The location on the surface of the earth is given by $\x = (\bi x,\psi(\bi x)) \in \mathbb{R}^3$, where $\bi x = (x_1,x_2) \in \Rb^2$ and $\psi: \mathbb{R}^2\rightarrow \mathbb{R}$ is a known ground topography.
$\tilde{\rho} : \mathbb{R}^2 \rightarrow \mathbb{R}$ denotes the ground reflectivity.
Table \ref{tab:notation} lists all major notation used throughout the paper.

\begin{figure}[!ht]
\centering
\includegraphics[scale=.45]{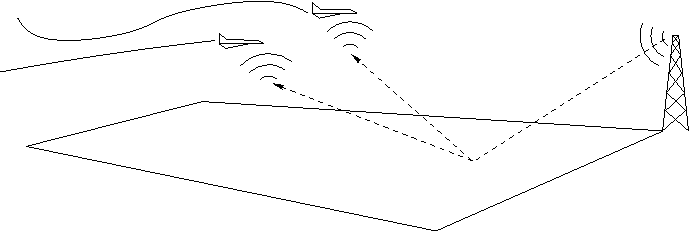}
\caption{An illustration of the generic passive SAR imaging configuration. The scene is illuminated by a stationary illuminator of opportunity. The backscattered signal is measured by two receivers.}
\label{fig:passive_sar_fig_chap4}
\end{figure}

\subsection{Received Signal Model}


Under the start-stop and Born approximations, assuming waves propagate in free space, the fast-time temporal Fourier transform of the received signal at the $i$-th receiver due to a stationary transmitter located at $\y$ can be modeled as \cite{Yarman08_biSAR}
\begin{equation}
\begin{split}
f_i(\omega,s) \approx \ & \Lc_i[\tilde{\rho}](\omega,s) := \int \mathrm{e}^{-\mathrm{i}\omega \phi_{i}(s,\bi x, \y)/c_0} \\ & \times A_i(\omega,s,\bi x)  \tilde{\rho}(\bi x) d\bi x, \ \ i=1,2,
\end{split}
\label{eq:received_signal}
\end{equation}
where
\begin{equation}
\phi_{i}(s,\bi x, \y) = |\y - \x| + |\x-\bgamma_{i}(s)|,
\label{eq:range}
\end{equation}
and
\begin{equation}
A_{i}(\omega,s,\bi x) = \frac{J_i(\omega,\bi x)}{(4\pi)^2 |\x-\bgamma_{i}(s)||\y-\x|}.
\label{eq:receive_A}
\end{equation}
In \eqref{eq:receive_A}, $J_i(\o,\bi x)$ models the antenna beam patterns and transmitted waveform.

\subsection{Model for Correlated Measurements}

We correlate the bistatic measurements acquired by two spatially separated receivers, given in \eqref{eq:received_signal}, as follows:
\begin{equation}
d_{12} = f_1(\omega,s)f_2^*(\omega,s)
\label{eq:correlation}
\end{equation}
where $f^*$ denotes complex conjugation.
Substituting \eqref{eq:received_signal} into \eqref{eq:correlation}, assuming the range of transmitter is much larger than that of the scene (i.e. $|\y|\gg |\x|$), we model the correlated measurements as follows \cite{Mason15}:
\begin{equation}
\begin{split}
d_{12}(\omega,s) \approx \ & \mathcal{F}_{12} [\rho](\omega,s) := \int
 \mathrm{e}^{-\mathrm{i}\omega \phi_{12}(s,\bi x,\bi x')/c_0} \\ & \times A_{12}(\omega,s,\bi x,\bi x') \rho(\bi x,\bi x') d\bi x d\bi x'
 \end{split}
\label{eq:forward_model_chap4}
\end{equation}
where
\begin{equation}
\rho(\bi x,\bi x') = \tilde{\rho}(\bi x)\tilde{\rho}^*(\bi x')
\end{equation}
and
\begin{equation}
\phi_{12}(s,\bi x,\bi x') = |\x - \bgamma_1(s)| - |\x' - \bgamma_2(s)| + \hat{\y}\cdot(\x - \x'),
\label{eq:ff_range_term1}
\end{equation}
\begin{equation}
A_{12}(\omega,s,\bi x,\x') = \frac{J_1(\omega,\bi x)J^*_2(\omega,\bi x')}{(4\pi)^2 |\x-\bgamma_{1}(s)||\x'-\bgamma_{2}(s)||\y|^2}.
\label{eq:forward_A}
\end{equation}
In \eqref{eq:ff_range_term1}, $\hat{\y}=\y/|\y|$ denotes the unit vector in the direction $\y$.

We define $\Rc = \tilde{\rho} \otimes \tilde{\rho}^*$ where $\otimes$ denotes the tensor product.
Clearly, $\Rc$ is a rank-one positive semi-definite (PSD) operator with kernel $\rho(\bi x,\bi x')$.
We refer to the kernel of $\Rc$ as the \textit{Kronecker scene}.
Since $\Rc$ is rank-one, the scene reflectivity is the only eigenfunction of $\Rc$, and can be recovered exactly from $\Rc$. 
While $d_{12}(\omega,s)$ is non-linear in scene reflectivity $\tilde{\rho}$, it is linear in $\rho$.
Therefore, our goal is to recover the \textit{Kronecker scene} $\rho$ using $\Fc_{12}[\rho] = d_{12}$.

In the rest of this paper, for notational brevity, we drop the $12$ subscripts from $d_{12}$ and $\mathcal{F}_{12}$. Furthermore, we represent the kernel of $\Fc$ as
\begin{equation}
F(\omega,s,\bi x,\bi x') = \mathrm{e}^{-\mathrm{i}\omega \phi(s,\bi x,\bi x')/c_0} A(\omega,s,\bi x,\bi x').
\end{equation}

\subsection{Discretized Model}

We discretize the scene reflectivity using the standard pixel basis and write
\begin{equation}
d(\omega,s) \approx \sum\limits_{k=1}^N \sum\limits_{k'=1}^{N} F(\omega,s,\bi x_k,\bi x_{k'})\rho(\bi x_k,\bi x_{k'})
\end{equation}
where $\rho(\bi x_k,\bi x_{k'})$ are the entries of an $N^2\times N^2$ rank-one positive semi-definite matrix, defined as $\brho$.
We vectorize this matrix by stacking the columns on top of each other, forming the vector $\bar{\brho}$.
The data is discretized by sampling uniformly in slow-time and fast-time to form a matrix $\mathbf{d}$ with entries $d(\omega_m,s_p), \ m=1,\dots,M, \ p=1,\dots,P$.
We then vectorize the data matrix by stacking the columns (each slow-time sample) in a vector, defined as $\bar{\mathbf{d}} \in \Cb^{MP}$.
The forward model is discretized using the same sampling schemes for the scene and data, this results in a matrix $\mathbf{F} \in \Cb^{MP\times N^4}$.
The columns and rows are ordered.
In particular, for the $mp$-th row, the $kk'$-th entry is given as
\begin{equation}
[\mathbf{F}_{mp}]_{kk'} = \mathrm{e}^{-\mathrm{i}\omega_m \phi_{12}(s_p,\bi x_k,\bi x_{k'}, \hat{\y})/c_0} A_{12}(\omega_m,s_p,\bi x_k,\bi x_{k'}).
\label{eq:discrete_kernel}
\end{equation}
 Thus, we write
 \begin{equation}
 \bar{\mathbf{d}}=\mathbf{F}\bar{\brho}.
\label{eq:discrete_data_model}
 \end{equation}
We use this discrete data model in applying the LRMR theory, and for numerical implementation.

\section{LRMR based Image Reconstruction}\label{sec:PA_LRMR_theory}

\subsection{Optimization based Imaging}

Since the unknown discretized Kronecker scene is a rank-one matrix we can pose image formation as the following rank minimization problem:
\begin{equation}
\mathrm{P}_1: \ \underset{\brho}{\text{minimize}} \ J_1(\brho) = \text{rank}(\brho) \quad \text{s.t.} \quad \mathbf{F}\bar{\brho} = \bar{\mathbf{d}}.
\label{eq:rank_min_chap4}
\end{equation}
This problem is non-convex and NP-hard.
The difficulty of this problem has motivated development of alternative heuristics, two of which apply to PSD matrices and are referred to as the log-determinant and trace heuristics \cite{fazel2003,Fazel02,Fazel01}.
More recently, within the context of phaseless imaging, a penalty with a higher degree of regularization that promotes a rank-one solution has been developed \cite{Bayram2017}.

While some of these alternative penalties may perform better as a result of the structure of our problem, we use the nuclear norm heuristic due to its theoretical connection to Problem $\mathrm{P}_1$ \cite{Recht10}.
This penalty is applicable to general rectangular matrices and leads to the following convex problem \cite{Recht10,Fazel01}:
\begin{equation}\label{eq:nuc_norm_min}
\mathrm{P}_2: \ \underset{\brho}{\text{minimize}} \ J_2(\brho) = \|\brho\|_* \quad \text{s.t.} \quad \mathbf{F}\bar{\brho} = \bar{\mathbf{d}},
\end{equation}
where $\|\cdot\|_*$ is the nuclear-norm\footnote{Also known as the Schatten-1 norm or trace norm.}, defined as
\begin{equation}
\| \brho \|_* = \sum_i \sigma_i(\brho),
\end{equation}
with $\sigma_i(\brho)$ being the singular values of $\brho$. From \eqref{eq:nuc_norm_min}, it is clear that the convex relaxation of the rank minimization problem is the $\ell_1$-norm of the singular values, which draws a direct analogy between cardinality minimization and rank minimization \cite{Oymak2011}.

For relatively small scale problems (matrices of size less than $100\times 100$), the nuclear-norm minimization problem in \eqref{eq:nuc_norm_min} can be reformulated as a semi-definite program and solved efficiently using interior point methods \cite{Recht10}.
However, for large scale problems, semi-definite programming fails and alternative algorithms must be used.
Since the nuclear norm is non-smooth, a popular and efficient algorithmic approach is to use forward-backward splitting methods \cite{Combettes2011}.

In our problem, since the unknown is positive semi-definite, an equivalent formulation of $\mathrm{P}_2$ is
\begin{equation}
\mathrm{P}_3: \ \underset{\brho}{\text{minimize}} \ J_3(\brho) = \text{Tr}[\brho] \quad \text{s.t.} \quad \mathbf{F}\bar{\brho}= \bar{\mathbf{d}}, \ \ \brho \succeq 0
\end{equation}
where $\brho \succeq 0$ denotes that $\brho$ is positive semi-definite, and $\text{Tr}[\cdot]$ denotes the trace operator.

To solve $\mathrm{P}_3$ we introduce Lagrange multipliers and use a saddle point method.
However, first we incorporate the PSD constraint to the objective function using the following indicator function
\begin{equation}
i_{C_+}(\brho) = \begin{cases}
& 0 \quad \ \brho \in C_+ \\
& \infty \quad \brho \notin C_+
\end{cases}
\end{equation}
where $C_+$ is the set of PSD matrices.
Thus, $\mathrm{P}_3$ becomes
\begin{equation}
\mathrm{P}_4: \ \underset{\brho}{\text{minimize}} \ J_4(\brho) = \text{Tr}[\brho] + i_{C_+}(\brho) \quad \text{s.t.} \quad \mathbf{F}\bar{\brho}= \bar{\mathbf{d}}.
\label{eq:final_prob}
\end{equation}
To obtain a simple iterative scheme, it is useful to enforce that the solution have minimum norm.
Thus, we modify $\mathrm{P}_4$ as
\begin{equation}
\mathrm{P}_5: \ \underset{\brho}{\text{minimize}} \ J_5(\brho) = \text{Tr}[\brho] + i_{C_+}(\brho) + \frac{1}{\lambda}\|\brho\|^2_F \quad \text{s.t.} \quad \mathbf{F}\bar{\brho}= \bar{\mathbf{d}}.
\label{eq:final_prob2}
\end{equation}
Next, we introduce a Lagrange multiplier $\bxi$ to incorporate the equality constraint into the objective function and obtain 
\begin{equation}
L(\brho,\bxi) = \text{Tr}[\brho] + i_{C_+}(\brho) + \frac{1}{\lambda}\| \brho \|^2_F + \langle \bxi , \mathbf{F}\bar{\brho} - \bar{\mathbf{d}} \rangle.
\label{eq:lagrange2}
\end{equation}
While $\mathrm{P}_5$ is no longer the same as $\mathrm{P}_4$, its solution is equal to that of $\mathrm{P}_4$ as $\lambda \rightarrow \infty$.
This follows from Theorem 3.1 in \cite{Cai2010}.
In our numerical simulations, we show that the solution to \eqref{eq:lagrange2} converges to the true solution for a sufficiently large $\lambda$.

Under strong duality assumption, the solution to \eqref{eq:final_prob2} can be obtained by solving the dual problem, using the Uzawa iterative procedure \cite{Cai2010}:
first update $\brho$ as the minimizer of \eqref{eq:lagrange2} for fixed $\bxi$, then update the Lagrange multiplier $\bxi$ with a gradient ascent update.
The final iterative scheme then becomes:
\begin{equation}
\begin{aligned}
\begin{cases}
\brho^{k} \ \ = & \Pc_{C_+} ( \text{Rs}(\mathbf{F}^H\bxi^k) - \lambda\mathbf{I} ) \\
\bxi^{k+1} = & \bxi^k + \beta_k (\mathbf{F}\bar{\brho}^k-\bar{\mathbf{d}})
\end{cases}
\end{aligned}
\label{eq:opt_alg}
\end{equation}
where $\Pc_{C_+}(\cdot)$ is a projection onto the PSD cone, $\beta_k$ is the step size, $\mathbf{I}$ is an identity matrix, $\bxi^0 = \mathbf{0}$, $\brho=\mathbf{0}$, and $\text{Rs}(\bar{\brho})=\brho$.

The convergence of \eqref{eq:opt_alg} to the solution of \eqref{eq:final_prob} is stated below in Theorem \ref{thm:converge}.
\begin{theorem}\label{thm:converge}
Suppose that the sequence of step sizes obeys $0 < \inf \beta_k \leq \beta_k \leq \sup \beta_k < 1/\alpha$, where $\alpha$ is the Lipshitz constant\footnote{Note that the $\alpha$ can be calculated as the largest singular value of $\mathbf{F}$.} of $f(\brho) = \mathbf{F}\bar{\brho}-\bar{\mathbf{d}}$. Then the sequence $\{ \brho^k \}$ generated by \eqref{eq:opt_alg} converges to the unique solution of $\mathrm{P}_5$, defined in \eqref{eq:final_prob}.
\end{theorem}

\begin{proof}
See Appendix \ref{sec:converge_proof_sketch}.
\end{proof}

\subsection{Recovery Guarantees}

We now provide the LRMR theory that states the conditions under which the rank-one solution to $\mathrm{P}_1$ is unique and the same as the solution to $\mathrm{P}_2$ \cite{Recht10}.
These recovery guarantees are developed under the assumption that the measurement matrix or forward model satisfies
the RIP generalized to low-rank matrices given below \cite{Recht10}.


\begin{definition}\textbf{Restricted Isometry Property}\label{def:rip_cond}
Let $\Ac: \Rb^{m\times n} \rightarrow \Rb^p$ be a linear operator. Assume, without loss of generality, $m<n$. For every $1 \leq r \leq m$, define the $r$-restricted isometry constant to be the smallest $\delta_r$ depending on $\Ac$, such that
\begin{equation}
(1-\delta_r)\|\mathbf{X}\|_F \leq \| \Ac[\mathbf{X}] \|_2 \leq (1+\delta_r)\|\mathbf{X}\|_F
\end{equation}
holds for all matrices $\mathbf{X}$, of rank of at most $r$, where $\|\mathbf{X}\|_F = \sqrt{\text{Tr}[\mathbf{X}^H\mathbf{X}]}$ is the Frobenius norm.
\end{definition}

The restricted isometry constant (RIC), $\delta_r$, can be viewed as a measure of how close a linear operator is to an isometry when restricted to a domain of matrices with rank at most $r$.
The relevant theorems generally state that when the RIC is less than unity there is a unique solution to the rank minimization problem, $\mathrm{P}_1$, and for an appropriately smaller $\delta_r$ this solution is the same as that of $\mathrm{P}_2$.

%

We begin our review of LRMR theory by considering the data model $\Ac[\mathbf{X}_0]=\mathbf{b}$ where $\text{rank}(\mathbf{X}_0)=r$.
Two recovery results, based on the assumption that the underlying forward model satisfyies the RIP condition with small enough RIC, are given in Theorems \ref{thm:recht_thm1} and \ref{thm:recht_thm2} \cite{Recht10}.

\begin{theorem}\label{thm:recht_thm1}
Let $\Ac(\mathbf{X}_0)=\mathbf{b}$ and suppose $\delta_{2r} < 1$ for some integer $r = \mathrm{rank}(\mathbf{X}_0) \geq 1$. Then, $\mathbf{X}_0$ is the only matrix of rank at most $r$ satisfying the matrix equation $\Ac(\mathbf{X}) = \mathbf{b}$, and hence the only solution of $\mathrm{P}_1$.
\end{theorem}

\begin{proof}
See proof of Theorem 3.2 in \cite{Recht10}.
\end{proof}

\begin{theorem} \label{thm:recht_thm2}
Suppose $\mathrm{r} = rank(\mathbf{X}_0) \geq 1$ and $\delta_{5r} \leq 1/10$. Then, $\mathbf{X}_0$ is the unique
solution to $\mathrm{P}_2$.
\end{theorem}

\begin{proof}
See proof of Theorem 3.3 in \cite{Recht10}.
\end{proof}

Theorem \ref{thm:recht_thm1} states that for $\delta_{2r}<1$, the solution to $\mathrm{P}_1$ in \eqref{eq:rank_min_chap4} is unique and equal to the exact solution $\mathbf{X}_0$.
Theorem \ref{thm:recht_thm2} states that for $\delta_{5r} \leq 1/10$, $\mathbf{X}_0$ is the unique and exact solution to $\mathrm{P}_2$ in \eqref{eq:nuc_norm_min}.
Since $\delta_r < \delta_{r'}$ when $r<r'$, Theorem \ref{thm:recht_thm2} implies Theorem \ref{thm:recht_thm1}.
From this, we conclude that when $\mathbf{F}$ satisfies RIP with $\delta_5<1/10$, the solutions of $\mathrm{P}_1$ and $\mathrm{P}_2$ are the same and equal to the true Kronecker scene $\brho^{\star}$.

In \cite{Chai11} the RIC condition in Theorem \ref{thm:recht_thm2} is improved and exact recovery is shown for a rank-one matrix when $\delta_2<1$.
We summarize this result in the following Theorem.

\begin{theorem}\label{thm:chai_thm}
Let $\Ac(\mathbf{X}_0)=\mathbf{b}$ with $\mathrm{rank}(\mathbf{X}_0) = 1$. Assume that $\delta_2 < 1$. Then $\mathbf{X}_0$ is the 
unique solution to $\mathrm{P}_2$. 
\end{theorem}

\begin{proof}
See proof of Theorem 3.4 in \cite{Chai11}.
\end{proof}

\section{Performance Guarantees for LRMR based Passive SAR Imaging}\label{sec:PA_analysis_gaurantees}

In this section, we prove an asymptotic recovery result for LRMR-based passive SAR. 
In summary, our result shows that exact reconstruction is possible if the target spacing is larger than a fraction of the wavelength corresponding to center frequency. Lemma \ref{thm:int_est} states the precise conditions required for exact recovery. 
Theorem \ref{thm:gaurantee} then shows that the RIC is small enough to satisfy the conditions stated in Theorems \ref{thm:recht_thm1}, \ref{thm:recht_thm2} and \ref{thm:chai_thm}, proving that the solution to $\mathrm{P}_2$ is unique and the same as the solution to $\mathrm{P}_1$.


To obtain RIC, we start with evaluating $\|\mathbf{F}\brho\|_2^2$: 
\begin{equation}\label{eq:l2_norm_data}
\begin{aligned}
\|\mathbf{F}\brho\|^2_2 = & \sum\limits_{m=1}^{M}\sum\limits_{p=1}^{P} \left\vert \sum\limits_{k=1}^N \sum_{k'=1}^N F(\omega_m,s_p,\bi x_k,\bi x_{k'})\rho(\bi x_k,\bi x_{k'})\right\vert^2 \\
= & \sum\limits_{k,k'=1}^{N^2} \sum\limits_{l,l'=1}^{N^2} \left( \sum\limits_{m=1}^{M}\sum\limits_{p=1}^{P}  F(\omega_m,s_p,\bi x_k,\bi x_{k'}) \right. \\ & \qquad \times \left. F^*(\omega_m,s_p,\bi x_l,\bi x_{l'}) \right) \rho(\bi x_k,\bi x_{k'}) \rho^*(\bi x_l,\bi x_{l'}).
\end{aligned}
\end{equation}
In order to determine RIC, we seek an estimate of the inner summation,
\begin{equation}\label{eq:inner_sum}
K(\bar{\bi x}_{k},\bar{\bi x}_{l}) := \sum\limits_{m=1}^{M}\sum\limits_{p=1}^{P}  F_{12}(\omega,s,\bar{\bi x}_k) F^*_{12}(\omega,s,\bar{\bi x}_l)
\end{equation}
where we define $\bar{\bi x}_k = [\bi x_k,\bi x_{k'}]$ and $\bar{\bi x}_l = [\bi x_l,\bi x_{l'}]$.

Assuming sufficiently many samples in fast-time frequency and slow-time,  we use \eqref{eq:forward_model_chap4} to approximate \eqref{eq:inner_sum}, as follows\footnote{In \eqref{eq:inner_int_2}, we apply a change of variables $\omega \rightarrow \omega - \omega_c$ and evaluate the $\omega$ integral in baseband $\Omega$.}:
\begin{equation}
K(\bar{\bi x}_k,\bar{\bi x}_l) \approx \int_{S\times\Omega} e^{-i(\omega_c + \omega) \theta(s,\bar{\bi x}_k,\bar{\bi x}_l)/c_0} B(\omega,s,\bar{\bi x}_k,\bar{\bi x}_l)  ds d\omega
\label{eq:inner_int_2}
\end{equation}
where
\begin{equation}
\begin{aligned}
\theta(s,\bar{\bi x}_k,\bar{\bi x}_l) = & \ |\x_k-\bgamma_1(s)|-|\x_{k'}-\bgamma_2(s)| -|\x_l-\bgamma_1(s)| \\  & \ +|\x_{l'}-\bgamma_2(s)| +\hat{\y}\cdot(\x_{k'}-\x_{k}-\x_{l'}+\x_{l})
\label{eq:theta_2}
\end{aligned}
\end{equation}
and
\begin{equation}
B(\omega,s,\bar{\bi x}_k,\bar{\bi x}_l) = \frac{1}{(4\pi)^4|\y|^4}B_1(s,\bar{\bi x}_k,\bar{\bi x}_l)B_2(\omega,\bar{\bi x}_k,\bar{\bi x}_l)
\label{eq:normal_amp}
\end{equation}
with
\begin{equation}
\begin{aligned}
B_1(s,\bar{\bi x}_k,\bar{\bi x}_l) = & \frac{1}{|\x_k-\bgamma_1(s)||\x_{k'}-\bgamma_2(s)|} \\ & \times \frac{1}{|\x_l-\bgamma_1(s)||\x_{l'}-\bgamma_2(s)|},
\end{aligned}
\label{eq:B1_1}
\end{equation}
and
\begin{equation}
B_2(\omega,\bar{\bi x}_k,\bar{\bi x}_l) = J_1(\omega,\bi x_k)J^*_2(\omega,\bi x_{k'})J^*_1(\omega,\bi x_l)J_2(\omega,\bi x_{l'}).
\label{eq:B2_1}
\end{equation}



In order to obtain RIC, we need to estimate the integral in \eqref{eq:inner_int_2}. To do so, we evaluate \eqref{eq:inner_int_2} for those $\bar{\x}_k,\bar{\x}_l$ where $\theta(s,\bar{\bi x}_k,\bar{\bi x}_l)$ is zero for all $s$ and for those where $\theta(s,\bar{\bi x}_k,\bar{\bi x}_l)$ is non-zero for almost all $s$ using the method of stationary phase. However, before we can estimate the integral, it is necessary to obtain a lower bound on $\theta(s,\bar{\bi x}_k,\bar{\bi x}_l)$. This lower bound requires a minimum achievable resolution stated in Lemma \ref{thm:suff_cond}.

We define the set of all possible pixel locations as $I = \{ (k,k',l,l') \in \Zb_+ \ \vert \ (\x_k,\x_{k'},\x_l,\x_{l'}) \in [-L_x/2,L_x/2]\times[-L_x/2,L_x/2] \}$. We then divide this set into two disjoint sets $I_1 = \{ (k,k',l,l') \in \Zb_+ \ \vert \ k=l \ \mathrm{and} \ k'=l' \}$ and $I_2 = \{ (k,k',l,l') \in \Zb_+ \ \vert \ k\neq l \ \mathrm{or} \ k'\neq l' \}$.  Specifically, these two sets satisfy $I = I_1 \cup I_2$ and $I_1 \cap I_2 = \emptyset$. Additionally, $\theta(s,\bar{\bi x}_k,\bar{\bi x}_l)$ is zero for all indices in $I_1$ and for all $s$. For each index in $I_2$, $\theta(s,\bar{\bi x}_k,\bar{\bi x}_l)$ is zero for at most two values of $s$ \footnote{This can be seen by setting $\theta(s,\bar{\bi x}_k,\bar{\bi x}_l) = 0$ over $I_2$ and solving for the roots of the resulting second degree polynomial.}. Hence, $\theta(s,\bar{\bi x}_k,\bar{\bi x}_l)$ is non-zero over $I_2$ for almost all $s$.

By studying the behaviour of \eqref{eq:inner_int_2} restricted to each of these sets, we obtain the results in Lemmas \ref{thm:suff_cond} and \ref{thm:int_est}.

\begin{lemma}\label{thm:suff_cond}
Assume the ground topography is flat\footnote{The result can be extended to non-flat topography. However, for simplicity of argument only flat topography case is considered.} and $(k,k',l,l') \in I_2$. Let $\Delta_x = \min_{k,l} |\x_k-\x_l|$ denote the minimum spacing between two scatterers, and $\alpha_{\y}$ be the angle between $\y$ and its orthogonal projection onto the ground plane.
Assume
\begin{equation}
\Delta_x \gg \frac{c_0}{\omega_c |\cos\alpha_{\y} - 1|}.
\label{eq:resolution}
\end{equation}
Then, $\theta(s,\bar{\bi x}_k,\bar{\bi x}_l)$ in \eqref{eq:theta_2} satisfies
\begin{equation}
\min\limits_{(k,k',l,l')\in I_2}\left\{ \frac{\omega_c}{c_0}\big\vert\theta(s,\bar{\bi x}_k,\bar{\bi x}_l)\big\vert \right\} \gg 1.
\label{eq:nec_bnd_2}
\end{equation}
\end{lemma}

\begin{proof}
See Appendix \ref{sec:suff_cond_prf}.
\end{proof}

\begin{lemma}\label{thm:int_est}
Suppose
\begin{equation}
C_J = \int_\Omega B_2(\omega,\bar{\bi x}_k,\bar{\bi x}_k) d\omega  < +\infty
\label{eq:cj_def}
\end{equation}
and $C_J > 0$.
Let
\begin{equation}
W(s,\bar{\bi x}_k,\bar{\bi x}_l) = \int\limits_{\Omega} e^{-i\frac{\omega}{c_0}\theta(s,\bar{\bi x}_k,\bar{\bi x}_l)} B_2(\omega,\bar{\bi x}_k,\bar{\bi x}_l) d\omega.
\label{eq:W_def}
\end{equation}
Assume 
$W(s,\bar{\bi x}_k,\bar{\bi x}_l)$ satisfies
\begin{equation}
\underset{s,\bar{\bi x}_k,\bar{\bi x}_l \in R}{\sup} \ \big\vert \partial^{\alpha}_{\bar{\bi x}_k} \partial^{\beta}_{\bar{\bi x}_l} \partial^{\eta}_{s} W(s,\bar{\bi x}_k,\bar{\bi x}_l) \big\vert \leq C_s (1+s^2)^{(2-|\eta|)/2}
\label{eq:mosp_assump}
\end{equation}
where $R$ is a compact subset of $S \times \Rb^2 \times \Rb^2$ and $C_s$ is a constant that depends on $R,\alpha,\beta,\eta$.
Furthermore, assume the conditions of Lemma \ref{thm:suff_cond} hold and the ranges of the receivers to the scene are much larger than the size of the scene $L_x$.
 Then, as $\omega_c \rightarrow \infty$,  
\begin{equation}
K(\bar{\bi x}_k,\bar{\bi x}_l) \approx
\begin{cases}
\frac{C_J \Delta_s}{(4\pi)^4|\y|^4 L_{g_1}^2 L_{g_2}^2}  & (k,k',l,l') \in I_1, \\
\begin{aligned}
 & \frac{1}{(4\pi)^4|\y|^4 L^2_{g_1}L^2_{g_2}} \left\{ \sqrt{\frac{2\pi}{\omega_c}} \right. \\ & \times \left. C_K + \Oc\left( \frac{1}{\omega_c^{3/2}} \right) \right\} 
 \end{aligned}  &  (k,k',l,l') \in I_2,
\label{eq:K_est}
\end{cases}
\end{equation}
where
$L_{g_i}$ is the range of the $i^{th}$ receiver from the center of the synthetic apertures and
\begin{equation}
\begin{aligned}
C_K = & \sum\limits_{\{ s_0 \vert \dot{\theta}(s_0,\bar{\bi x}_k, \bar{\bi x}_l)=0 \}} \frac{e^{i\omega_c\theta(s_0,\bar{\bi x}_k, \bar{\bi x}_l)}e^{i\pi/4\ddot{\theta}(\bar{\bi x}_k, \bar{\bi x}_l,s_0)}}{\sqrt{|\ddot{\theta}(s_0,\bar{\bi x}_k, \bar{\bi x}_l)|}} \\ & \qquad \qquad \qquad \quad \times W(s_0,\bar{\bi x}_k, \bar{\bi x}_l),
\end{aligned}
\end{equation}
with $s_0$ satisfying $\dot{\theta}(s_0,\bar{\bi x}_k,\bar{\bi x}_l) = 0$ and $\ddot{\theta}(s_0,\bar{\bi x}_k,\bar{\bi x}_l) \neq 0$.
\end{lemma}

\begin{proof}
See Appendix \ref{sec:int_est_pf}.
\end{proof}


These two lemmas are necessary to obtain exact reconstruction. The result of Lemma \ref{thm:suff_cond} is required to ensure that we are in the asymptotic regime allowing us to use the method of stationary phase in proving Lemma \ref{thm:int_est}. The result of Lemma \ref{thm:suff_cond}, in return, requires a minimum target spacing, yielding a lower bound on the achievable resolution for exact reconstruction.

Under the assumptions of Lemmas \ref{thm:suff_cond} and \ref{thm:int_est}, we can now obtain the following performance guarantee.

\begin{theorem}\label{thm:gaurantee}
Suppose the conditions of Lemma \ref{thm:suff_cond} and \ref{thm:int_est} hold. Then as $\omega_c \rightarrow \infty$, the exact locations and amplitudes of scatterers can be recovered by solving $\mathrm{P}_2$ in \eqref{eq:nuc_norm_min}.
\end{theorem}

\begin{proof}
See Appendix \ref{sec:gaurantee_pf}.
\end{proof}

\section{Phaseless Passive Imaging and Exact Recovery}\label{sec:phaseless}


When the two receivers are collocated, i.e., $\bgamma_1(s)=\bgamma_2(s)=\bgamma(s)$ the passive SAR forward model in \eqref{eq:forward_model_chap4} reduces to that of phaseless imaging.
In this case, the phase function in \eqref{eq:theta_2} becomes
\begin{equation}
\begin{aligned}
\theta(s,\bar{\bi x}_k,\bar{\bi x}_l) = & \ |\x_k-\bgamma(s)|-|\x_{k'}-\bgamma(s)| -|\x_l-\bgamma(s)| \\  & \ +|\x_{l'}-\bgamma(s)| +\hat{\y}\cdot(\x_{k'}-\x_{k}-\x_{l'}+\x_{l}).
\label{eq:theta_phaseless}
\end{aligned}
\end{equation}
We show that the uniqueness results obtained in \cite{Chai11} also follow from our analysis,
generalizing it to arbitrary imaging geometries.

In the previous section, we decompose the pixel locations into two sets $I_1$ and $I_2$ so that \eqref{eq:theta_2} is zero on $I_1$ and non-zero on $I_2$.
In the case of phaseless imaging, these sets are different.
More specifically, there is a set $I_3 \subset I_2$ for which the phase function \eqref{eq:theta_phaseless} is zero, but \eqref{eq:theta_2} is not.
We define this set as $I_3 = \{ (k,k',l,l') \in \Zb_+ \  k=k' \ \mathrm{or} \ l=l', k\neq l \}$, and let $\tilde{I}_2 = I_2 \backslash I_3$, and $\tilde{I}_1 = I_1 \cup I_3$. We carry out the same analysis in Section \ref{sec:PA_analysis_gaurantees} using the sets $\tilde{I}_1$ and $\tilde{I}_2$. However, for phaseless imaging it is useful to consider $I_1$ and $I_3$ separately, instead of $\tilde{I}_1$.

As in the previous case we can obtain recovery guarantees, however, the RIC is larger when using auto-correlated measurements than cross-correlated measurements.
This is a result of the fact that $\theta(s,\bar{\bi x}_k,\bar{\bi x}_l)=0$ in the additional set $I_3$ in phaseless imaging

Before we state the recovery results for phaseless imaging, we briefly comment on the proofs of Lemmas \ref{thm:suff_cond} and \ref{thm:int_est} with $I_1$ and $I_2$ replaced with $\tilde{I}_1$ and $\tilde{I}_2$, respectively.
First, the proof of Lemma \ref{thm:suff_cond} remains the same, since in this case we only need to consider the difference between the distance of two different scatterer locations to the receivers. This is the case, because the only way three of the four magnitude terms in the phase cannot be equal is for one them to be equal to the fourth, and this would make the phase function zero.
Thus, the approach to the proof and the resulting final condition follow from the same procedure as the proof of Lemma \ref{thm:suff_cond}.
In the case of Lemma \ref{thm:int_est}, the proof remains exactly the same, despite changing the sets.
Therefore, we can use these two lemmas to now state the main recovery result for phaseless imaging.

\begin{theorem}\label{thm:phaseless_thm}
Suppose $\bgamma_1(s) = \bgamma_2(s) = \bgamma(s)$ for all $s\in S$. Let $\tilde{I}_1 = \{ (k,k',l,l') \in \Zb | (k=k' \ \mathrm{and} \ l=l') \ \mathrm{or} \ (k=l \ \mathrm{and} \ k'=l', k \neq k') \}$ and $\tilde{I}_2 = \{ (k,k',l,l') \in \Zb | (k \neq k' \ \mathrm{or} \ l\neq l') \ \mathrm{and} \ (k \neq l \ \mathrm{or} \ k \neq  l') \}$. Suppose the conditions in Lemma \ref{thm:suff_cond} and \ref{thm:int_est} hold with $I_1$ and $I_2$ replaced with $\tilde{I}_1$ and $\tilde{I}_2$, respectively. Then as $\omega_c \rightarrow \infty$, the exact locations and amplitudes of the scatterers can be recovered by solving $\mathrm{P}_4$.
\end{theorem}

\begin{proof}
See Appendix \ref{sec:phaseless_gaurantee_pf}.
\end{proof}

\section{Discussion}\label{sec:results_discuss}

We begin our discussion with the connection of the algorithmic approach in \eqref{eq:opt_alg} to the recovery results stated in Theorems \ref{thm:gaurantee} and \ref{thm:phaseless_thm}.
The recovery results state that when the measurement map satisfies RIP with sufficiently small RIC, the solutions of $\mathrm{P}_4$ and $\mathrm{P}_1$ are unique and equal.
This implies that the exact Kronecker scene can be recovered. 
While $\mathrm{P}_4$ can be solved efficiently using interior point methods, this approach is not suitable for the large scale problem of SAR imaging.
Therefore, in order to solve $\mathrm{P}_4$ efficiently we propose the use of $\mathrm{P}_5$, which is strictly convex and has unique solution.
For sufficiently large $\lambda$, this solution is the minimum Frobenius-norm solution of $\mathrm{P}_4$ which has unique solution for sufficiently large $\omega_c$.
Therefore, when $\lambda$ and $\omega_c$ are sufficiently large, the iteration in \eqref{eq:opt_alg} converges to the true Kronecker scene and thus recovers the scene reflectivity exactly.

Along with provable exact reconstruction, our analysis provides a condition on the minimum resolvable target spacing. 
The condition is stated in Lemma \ref{thm:suff_cond} and suggests that convex LRMR based methods have \qq{super-resolution} capability, in the sense that it offers superior resolution than that of Fourier-based imaging.
While the resolution condition depends on $\omega_c$ and $\alpha_{\y}$, which are not user controlled parameters in passive SAR imaging, for typical terrestrial illuminators of opportunity the values of these parameters are such that super-resolution can be achieved.
This is an appealing attribute of our method as resolution does not depend on bandwidth as in the case of Fourier-based methods.
For example, DVB-T and WiMAX waveforms have bandwidth of $6$-$8$ MHz and operate at around $760$ MHz and $2$ GHz center frequencies, respectively.
If the transmitter elevation angle is $15$ degrees, the resolution lower bound is $2$m for DVB-T and even smaller for WiMAX.
This is far superior than the $19$-$20$m of resolution achievable by Fourier-based imaging.

By comparing the results of Theorems \ref{thm:gaurantee} and \ref{thm:phaseless_thm}, it is clear that using cross-correlated measurements results in a smaller RIC than using auto-correlated measurements.
This implies that $\omega_c$ required for exact reconstruction is smaller in the case of cross-correlated measurements.
Smaller RIC is due to the fact that $I_1 \subset \tilde{I}_1$, i.e., the phase of the forward map is zero on a larger set of pixels in phaseless imaging.
Hence, we conclude that a smaller RIC results from retaining phase information in the cross-correlated data.
While the recovery guarantees are more easily obtained by cross-correlating measurements, auto-correlating offers incoherent imaging capability.
This is useful when there is synchronization error, a common problem in radar imaging using distributed apertures \cite{Hack2014_2}.

We conclude this section with a discussion of the validity of the assumptions required in proving the recovery results.
The main assumptions are made in evaluation of \eqref{eq:inner_int_2}, stated in Lemma \ref{thm:int_est}. 
The required assumptions apply to the imaging geometry, transmitted waveform, and antenna beam patterns.
With respect to the imaging geometry, we require that the ranges of the antennas are much smaller than the size of the scene and synthetic aperture.
This assumption is satisfied by all typical airborne and space-borne imaging geometries \cite{Casteel2007}.
With respect to assumptions \eqref{eq:cj_def} and \eqref{eq:W_def}, they are satisfied when the antennas are sufficiently broadband and the transmitted waveform has flat spectrum or constant modulus.
These are likely to hold as most antennas are designed to be broadband, and many illuminators of opportunity transmit waveforms with flat spectrum, such as DVB-T and WiMAX signals.

\section{Numerical Simulations}\label{sec:numerical_simulations}

In this section we validate the recovery result presented in the previous sections.

\subsection{Experimental Set-up}\label{sec:exp_set_up}

We consider a scene with flat topography that is $[0,110]\times[0,110]\mathrm{m}^2$ discretized into $11\times 11$ pixels. 
This results in a Kronecker scene of $121 \times 121$ pixels, each pixel having $10\mathrm{m}$ of resolution.
The scene consists of a single extended target with three different reflectivity values as shown in Figure \ref{fig:eig_img_orig}.
The corresponding Kronecker scene is shown in Figure \ref{fig:kron_img_orig}.

\begin{figure*}[!t]
\centering
\subfigure[]{
	\includegraphics[width=0.45\textwidth]{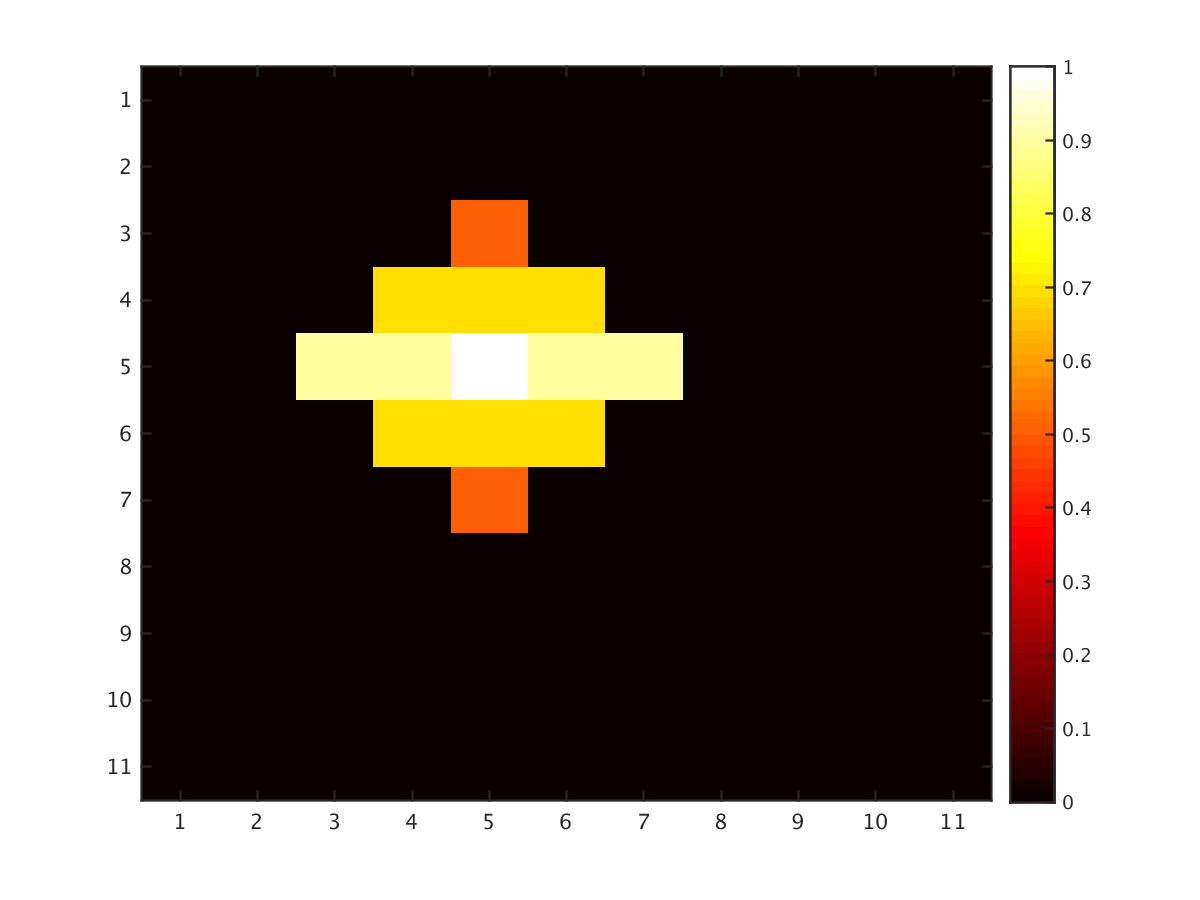}
	\label{fig:eig_img_orig}
}
\subfigure[]{
	\includegraphics[width=0.45\textwidth]{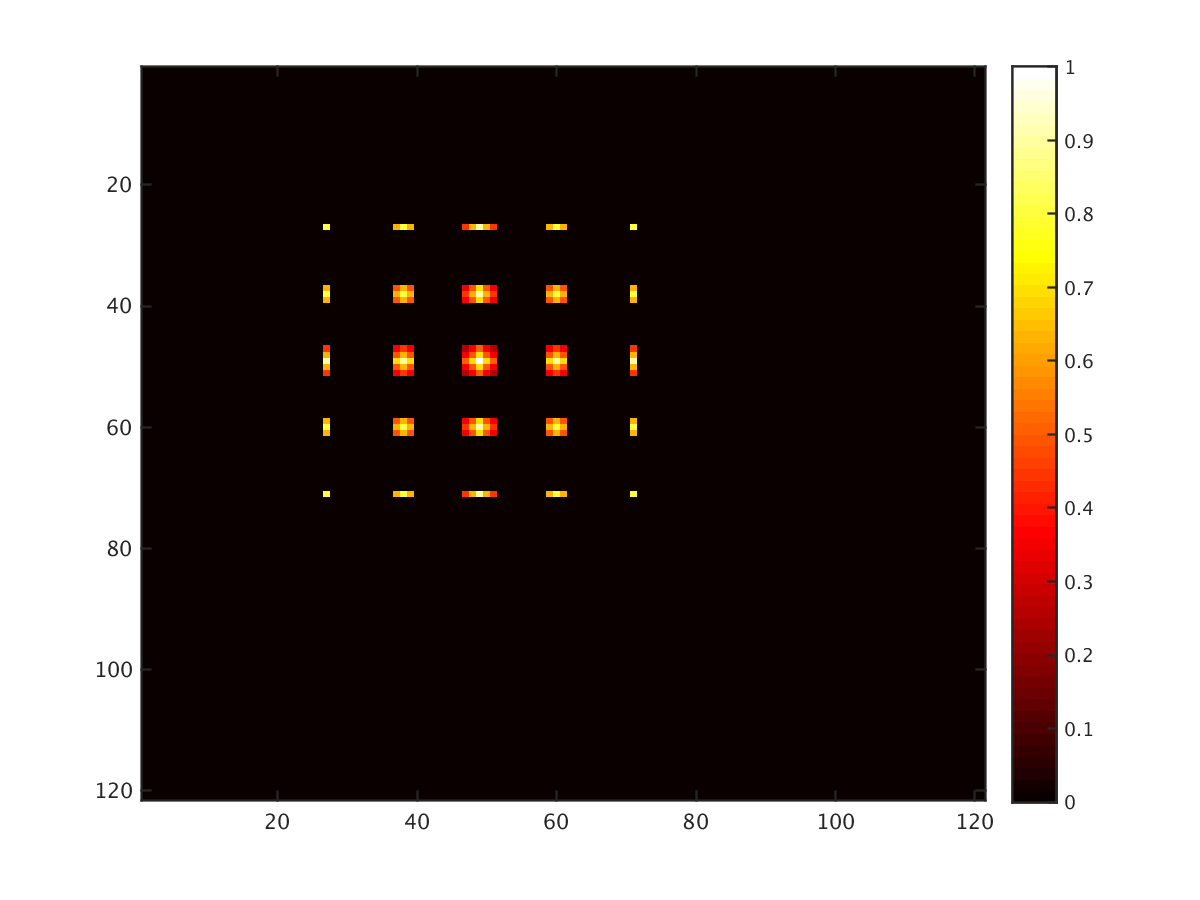}
	\label{fig:kron_img_orig}
}
\caption{ The scene reflectivity and Kronecker scene phantom used in our simulations. \subref{fig:eig_img_orig} Scene reflectivity. \subref{fig:kron_img_orig} Kronecker scene. }
\label{fig:img_phantoms}
\end{figure*}


We assume the scene is illuminated by a stationary transmitter located at $[12,12,5]\mathrm{km}$. 
For flat topography, this results in a transmitter elevation angle $\alpha_{\y} = 16.4832$ degrees.
We use system parameters corresponding to common illuminators of opportunity.
 Specifically, we assume a transmit pulse with a flat spectrum and $8$ MHz bandwidth, and center frequencies of $760$ MHz and $2$ GHz, corresponding to DVB-T and WiMAX signals, respectively \cite{Palmer13,Arroyo13}. 
We consider a semi-circular aperture, i.e., $\bgamma_1(s) = [7\cos(s),7\sin(s),5]\mathrm{km}, \ s\in[0,\pi/2]$, where the second receiver traverses the same trajectory with a constant offset of $\pi/4$, i.e., $\bgamma_2(s)=\bgamma_1(s+\pi/4)$.
An illustration of this imaging configuration is displayed in Figure \ref{fig:imaging_config}.
Based on this configuration, the resolution lower bound in \eqref{eq:resolution} is $1.53$m$^2$ and $0.58$m$^2$ for the $760$ MHz and $2$ GHz centrer frequencies, respectively. 
Thus, the $10$m$^2$ pixel spacing used in our experiments satisfies the resolution condition.

\begin{figure}[!t]
\centering
\includegraphics[width=0.45\textwidth]{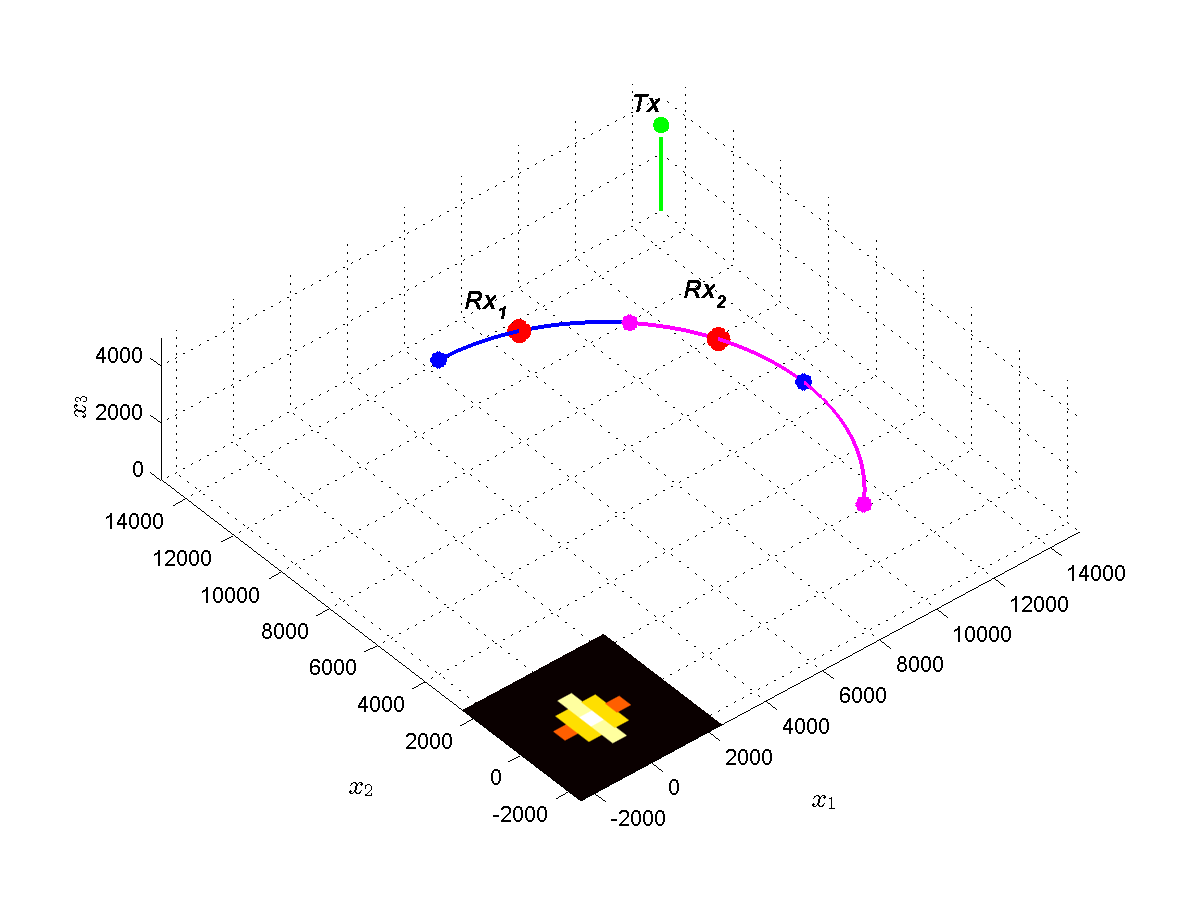}
\label{fig:circ_scene}
\caption{ The imaging geometry used in the simulations, depicting two receivers $text{Rx}_1,\text{Rx}_2$ depicted as red dots, traversing the same semi-circular trajectory with a constant $\pi/4$ offset. The trajectory of the first receiver is displayed in blue, and the end points marked by blue dots. The second receiver trajectory is marked in maroon and the end points marked by maroon dots. The stationary transmitter $\text{Tx}$ is indicated by a green dot. 
The scene is small and located at the center of the coordinate system (SCENE NOT DRAWN TO SCALE). 
}
\label{fig:imaging_config}
\end{figure}

 We assume isotropic transmit and receive antennas, and for simplicity assume that geometric spreading has been compensated for during data collection.  
  This results in the amplitude function of the forward model in \eqref{eq:forward_A} being set to $1$. 
  The data at each receiver are simulated using \eqref{eq:received_signal}, then correlated according to \eqref{eq:correlation}. 

%

\subsection{Experimental Results}

\subsubsection{Performance Metrics and Evaluation Criteria}

We denote the true and reconstructed Kronecker scene as $\brho$ and $\brho^{\star}$, and evaluate the reconstruction performance to validate exact recovery ($\brho=\brho^{\star}$) using three error metrics:
\begin{equation}
\label{eq:error_metrics}
\begin{split}
E_{\mathbf{d}}(\brho) & = \frac{\| \mathbf{d} - \mathbf{F}\bar{\brho} \|_2}{\|\mathbf{d}\|_F} \qquad E_{\brho}(\brho) = \frac{\| \brho - \brho^{\star} \|_F}{\|\brho\|_F} \\ & \qquad \quad E_{\tilde{\brho}}(\tilde{\brho}) = \frac{\| \tilde{\brho} - \tilde{\brho}^{\star} \|_F}{\|\tilde{\brho}\|_F}
\end{split}
\end{equation}
where $\|\cdot\|_F$ stands for the Frobenius norm.
These metrics evaluate how well the data fidelity constraint is satisfied, and how close the Kronecker scene and scene reflectivity are to the actual values.
We consider success for values smaller than $0.05 \%$ error with correct trace and rank. 
The error metrics show that we have obtained a feasible solution, and exact trace and rank confirm that we have reached a minimizer of $\mathrm{P}_4$, which is equivalent to that of $\mathrm{P}_1$. 
These conditions also verify that realistic center frequencies satisfy the assumptions of the recovery Theorems.

We run simulations for both cross-correlated and auto-correlated data.
The auto-correlated data corresponds to passive phaseless imaging. 
We run the algorithm stated in \eqref{eq:opt_alg} for $5000$ iterations with $\lambda = 20$ to recover $\brho$.
While the algorithm may not have completely converged in all cases, we consider a longer running time to be impractical.
We then reconstruct the scene reflectivity from $\tilde{\brho}$ by keeping the eigenvector corresponding to the largest eigenvalue. 

\subsubsection{Reconstruction Results}

\begin{table*}[!t]
\centering
\begin{tabular}{|c|c|c|c|c|c|c|}
\hline
$f_c$ & Type  & Tr & rank & $E_{\mathbf{d}}$ & $E_{\brho}$ & $E_{\tilde{\brho}}$  \\ \hline\hline
\multirow{ 2}{*}{$760$ MHz} & Cross & $7.6800$ & $1$ & $1.4448e-4$ & $4.8579e-5$ & $4.1516e-5$ \\\cline{2-7}
& Auto  & $7.4284$ & $19$ & $0.0234$ & $0.9150$ & $0.8096$ \\ \hline\hline
\multirow{ 2}{*}{$2$ GHz} & Cross & $7.6800$ & $1$ & $1.4961e-4$ & $5.9425e-5$ & $5.0066e-5$ \\\cline{2-7}
& Auto & $7.1306$ & $16$ & $0.0297$ & $0.8806$ & $0.7043$  \\ \hline
\end{tabular}
\caption{Results for semi-circular aperture.}
\label{tab:semi_circ_tab}
\end{table*}


In Table \ref{tab:semi_circ_tab} we tabulate the values of the three error metrics in \eqref{eq:error_metrics}, as well as the trace and the rank of the reconstructed Kronecker scene.
The true trace of the Kronecker scene displayed in \ref{fig:kron_img_orig} is $7.68$.
We find that using cross-correlated measurements results in exact reconstruction for both center frequencies and the reconstructed images visually match the Kronecker scene and reflectivity phantom shown in Figure \ref{fig:img_phantoms}.
The rank and trace match that of the true scene, and all error metrics in \eqref{eq:error_metrics} are below the $0.05\%$ threshold.
From this we conclude that the optimal point of $\mathrm{P}_1$-$\mathrm{P}_5$ has been recovered, validating our analysis.

While we obtain exact recovery for cross-correlated measurements, the algorithm fails to reconstruct the exact solution from auto-correlated or phaseless measurements. 
This suggests that the center frequency is not large enough for Theorem \ref{thm:chai_thm} to hold. 
Unlike our numerical results, \cite{Chai_thesis,Chai11} demonstrate exact reconstruction.
This is explained by the fact that they are using $30$ GHz or $300$ GHz center frequencies, which are not realistic operating frequencies for passive radar.

\begin{figure*}[!t]
\centering
\subfigure[]{
	\includegraphics[width=0.45\textwidth]{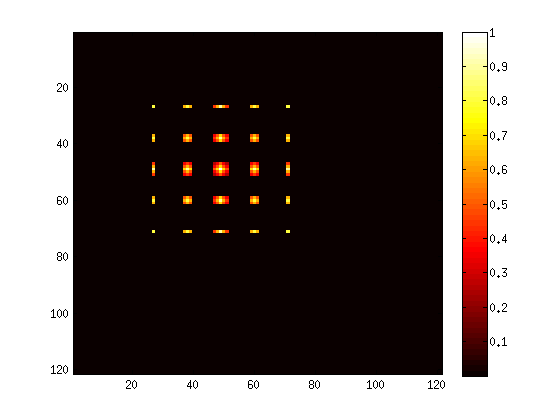}
	\label{fig:kron_img_cross}
}
\subfigure[]{
	\includegraphics[width=0.45\textwidth]{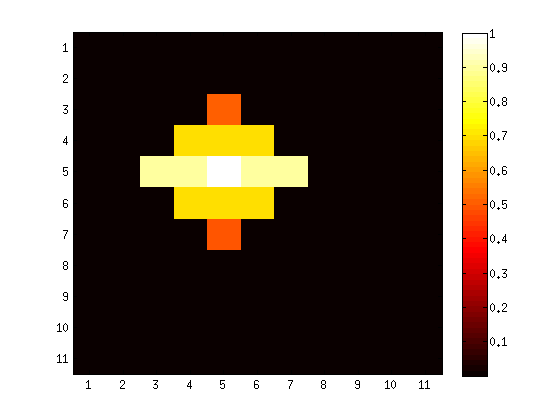}
	\label{fig:eig_img_cross}
}
\subfigure[]{
	\includegraphics[width=0.45\textwidth]{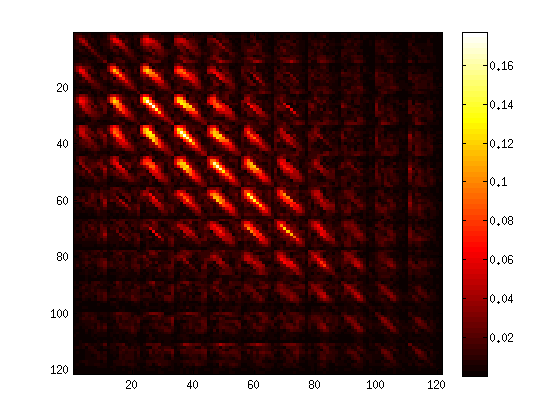}
	\label{fig:kron_img_auto}
}
\subfigure[]{
	\includegraphics[width=0.45\textwidth]{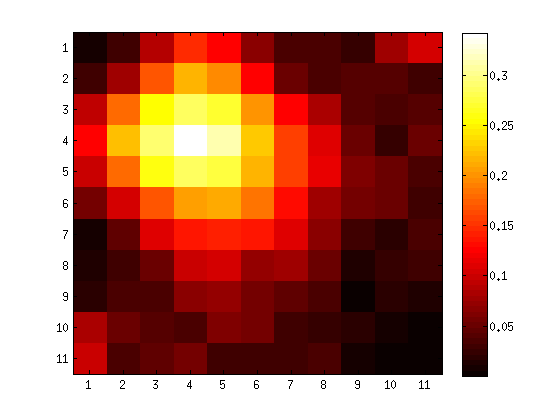}
	\label{fig:eig_img_auto}
}
\caption{ The reconstructed scene reflectivity and Kronecker scene phantom corresponding to the $2$ GHz center frequency experiments listed in Table \ref{tab:semi_circ_tab} using cross-correlated and auto-correlated measurements. \subref{fig:kron_img_cross} Reconstructed Kronecker scene using cross-correlated measurements. \subref{fig:eig_img_cross} Reconstructed scene reflectivity corresponding to the image in Figure \ref{fig:kron_img_cross}. \subref{fig:kron_img_auto} Reconstructed Kronecker scene using auto-correlated measurements. \subref{fig:eig_img_auto} Reconstructed scene reflectivity corresponding to the image in Figure \ref{fig:kron_img_auto}. }
\label{fig:recon_images}
\end{figure*}

In Figure \ref{fig:recon_images} we display the reconstructed images for $2$ GHz center frequency. Figures \ref{fig:kron_img_cross} and \ref{fig:kron_img_auto} display the reconstructed Kronecker scene using cross-correlated and auto-correlated measurements, respectively. 
The scene reflectivity corresponding to these images are shown in Figures \ref{fig:eig_img_cross} and \ref{fig:eig_img_auto}, respectively.
Visually, the reconstruction is exact when using cross-correlated measurements.
In the case of phaseless passive imaging, it is clear that the algorithm fails to converge to the correct solution. The resulting image is significantly blurred and the true target shape and reflectivity values are not recovered.


\subsubsection{Algorithm Behaviour}

\begin{figure}
\centering
\subfigure[]{
	\includegraphics[width=0.45\textwidth]{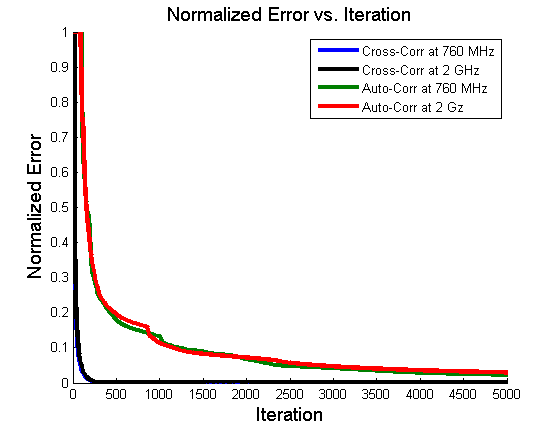}
	\label{fig:error_plot}
}
\subfigure[]{
	\includegraphics[width=0.45\textwidth]{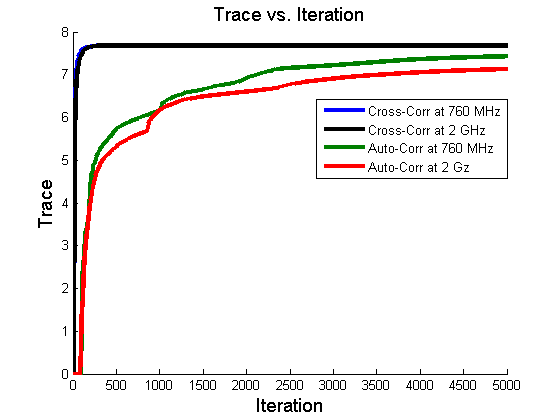}
	\label{fig:trace_plot}
}
\subfigure[]{
	\includegraphics[width=0.45\textwidth]{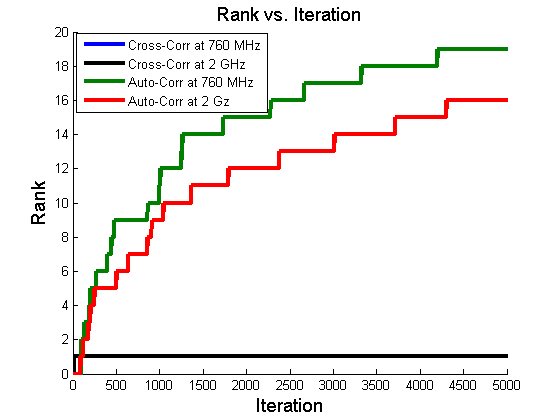}
	\label{fig:rank_plot}
}
\caption{ Plots illustrating algorithm performance for results displayed in Table \ref{tab:semi_circ_tab}. The cross-correlation results for $760$ MHz and $2$ GHz are displayed in blue and black, respectively. The auto-correlation results for for $760$ MHz and $2$ GHz are displayed in green and red, respectively. \subref{fig:error_plot} Error $E_{\mathbf{d}}$ verse iteration. \subref{fig:trace_plot} Trace verse iteration. \subref{fig:rank_plot} Rank verse iteration. }
\label{fig:plots}
\end{figure}

To demonstrate the convergence behavior of our method, we plot the recovered trace trace, rank, and $E_{\mathbf{d}}$ verses iteration number in Figure \ref{fig:plots}, for each experiment in Table \ref{tab:semi_circ_tab}.  
The plot of $E_{\mathbf{d}}$ is shown in Figure \ref{fig:error_plot}.
We observe that in every case the error decreases monotonically, and appears to be converging to zero.
When using phaseless measurements it appears that the algorithm has not fully converged, confirmed by the plots of trace in \ref{fig:trace_plot}. 
Also, based on the plot of trace and rank, it appears that the algorithm is converging to an incorrect solution. 
This is explained by the fact that the center frequency is not large enough for the result of Theorem \ref{thm:phaseless_thm} to hold.
This means the solution of $\mathrm{P}_4$ is not unique-and-equal to that of $\mathrm{P}_1$.
This is further verified by observing that the rank is converging to a value larger than one, which implies that the minimum norm solution of $\mathrm{P}_5$ is not rank-one.

Theorem \ref{thm:converge} states that the sequence of Kronecker scene estimates converge to the true solution, implying that $E_{\mathbf{d}}$ and trace converge to the correct value, as observed when using cross-correlated measurements.
Although, an interesting behavior is that the rank and trace increase monotonically while converging.
While this may seem odd at first, it is explained by the form of the algorithm.
Specifically, \eqref{eq:opt_alg} does not explicitly solve $\mathrm{P}_5$ by minimizing the trace of $\brho$, but instead solves the dual problem.
The solution of the dual problem lower bounds the Lagrangian at the optimal point, and since strong duality holds equality is obtained.
Therefore, at each iteration the trace of the iterate increases until the lower bound converges to the optimal value of the solution.

As stated before, at the $2$ GHz center frequency Theorem \ref{thm:phaseless_thm} does not hold.
Therefore, the solution of $\mathrm{P}_4$ and $\mathrm{P}_5$ are not the same as that of $\mathrm{P}_1$.
This implies that the solution to $\mathrm{P}_4$ and $\mathrm{P}_5$ may not be rank-one, in fact, the solution likely has rank greater than one.
Hence, the rank increases per iteration since the algorithm is converging to a solution with rank greater than one.
This is further verified by the fact that when cross-correlated measurements are used the result of Theorem 7 applies and the unique solution is rank-one. Thus, at each iteration the rank of the iterate remains at one.

This illustrates the difference in the behavior of the algorithm for cross- and auto-correlated measurements. 
It also offers a heuristical method to determine early on if the analysis holds given the imaging parameters.
If the rank increases above one, it can be assumed that the center frequency is not large enough and the algorithm can be terminated to save computational resources.
\section{Conclusion}\label{sec:conclusion}

In this work, we presented a performance analysis of the LRMR based convex passive SAR imaging method. 
Using the restricted isometry property we show that exact recovery is possible for sufficiently large center frequencies.
In evaluating the restricted isometry constant we determine a sufficient condition for exact recovery.
This condition arises from the application of the method of stationary phase. 
The condition states that the distance between scatterers must be larger than a certain bound that depends on the center frequency of the transmitted waveform and the elevation angle of the transmitter. For typical transmitters of opportunity operating in the $700$ MHz to $2$ GHz range, this condition translates into an achievable resolution that is an order of magnitude higher than that of Fourier based methods used in PCL and TDOA-based backprojection making convex LRMR-based passive SAR imaging a super-resolution technique.

Furthermore, we show that our analysis applies to phaseless passive imaging, in which measurements are auto-correlated instead of cross-correlated.
Our recovery results apply to phaseless imaging, however, with larger RIC and larger center frequency than those required by passive SAR with cross-correlated measurements.

To solve the imaging problem we present a projected gradient descent type algorithm and prove its convergence.
We verify our analysis with extensive numerical simulations.
We show that for typical operating frequencies used in passive radar, using cross-correlated measurements the scene reflectivity can be recovered exactly.

While we consider primarily delay based cross-correlated measurements, our methodology can be extended to study recovery results for Doppler-based correlations \cite{Yarman10,Wang11}.  Additionally, our results apply to passive synthetic aperture imaging acoustics and geophysics as well as passive imaging with sufficiently large distributed receivers \cite{LWang12,wang12,Wang10}.


\begin{appendices}

\section{Method of Stationary Phase}\label{sec:mosp}

The stationary phase theorem states that if u is a (possibly complex-valued) smooth function of compact support in $\Rb^n$, and $\varphi$ is a real-valued function with
only non-degenerate critical points, then as $\lambda \rightarrow \infty$
\begin{equation}
\begin{split}
\int e^{i\lambda\varphi(\bi x)} u(\bi x) d^n\bi x & =  \sum\limits_{\{ \bi x_0 \vert D\varphi(\bi x_0)=0 \}} \left( \frac{2\pi}{\lambda} \right)^{n/2} \\ & \times \frac{e^{i\lambda\varphi(\bi x_0)}e^{i\pi/4\text{sgn}D^2\varphi(\bi x_0)}}{\sqrt{|\text{det}D^2\varphi(\bi x_0)|}}u(\bi x_0) \\ & + \Oc(\lambda^{-n/2-1}).
\end{split}
\end{equation}
Here $D\varphi$ denotes the gradient of $\varphi$, and $D^2 \varphi$ denotes the Hessian, and sgn denotes the signature of a matrix, i.e. the number of positive eigenvalues minus the negative ones. 
Furthermore, a point $\bi x_0 \in \Rb^n$ is called a non-degenerate critical point
if $D\varphi(\bi x_0 ) = 0$ and the Hessian matrix $D^2 \varphi(\bi x_0 )$ has non-zero determinant.

\section{Sketch of Proof of Theorem \ref{thm:converge}}\label{sec:converge_proof_sketch}

Observe that in the noiseless case there exists a $\bar{\brho}$ such that $f(\brho)=0$, therefore, Slater's condition is satisfied and strong duality holds. 
This implies that the solution of \eqref{eq:opt_alg} is the unique minimizer of $\mathrm{P}_5$.
After establishing strong convexity of the objective function of $\mathrm{P}_5$, the proof follows the argument of Theorem 3.2 in \cite{Cai2010}.


Let $\brho,\brho' \in C_+$, where $C_+$ denotes the PSD cone. Then the derivative of $J_5$ at these points are $\mathbf{p} = \lambda\mathbf{I} + \mathbf{p}_0 + \brho$ and $\mathbf{p}' = \lambda\mathbf{I} + \mathbf{p}'_0 + \brho$, where $\mathbf{p}_0 \in  \partial i_{C_+}(\brho)$ and $\mathbf{p}'_0 \in \partial i_{C_+}(\brho')$ \footnote{ $v \in \partial g(u)$, denotes that $v$ is a subgradient of $g$ at $u$.}. 
Then,
\begin{equation}
\begin{aligned}
\langle \mathbf{p} - \mathbf{p}',\brho - \brho' \rangle =  \langle \mathbf{p}_0 - \mathbf{p}_0', \brho - \brho' \rangle + \| \brho - \brho' \|^2_F.
\label{eq:diff_product}
\end{aligned}
\end{equation}
From the definition of a subgradient of the PSD cone it can be shown that the first term on the right is  non-negative, and \eqref{eq:diff_product} becomes
\begin{equation}
\langle \mathbf{p} - \mathbf{p}',\brho - \brho' \rangle \geq \| \brho - \brho' \|^2_F,
\end{equation}
which is the definition of strong convexity with constant $1$.

Let $(\brho^{\star},\bxi^{\star})$ be a primal-dual optimal pair for problem $\mathrm{P}_5$.
From the first order optimality conditions, the strong convexity of $J_5$, and the Lipshitz continuity of $f(\brho) = \mathbf{d}  - \mathbf{F}\brho$, the sequence of Lagrange multiplier iterates obeys:
\begin{equation}
\| \bxi^k - \bxi^{\star} \|^2_2 \leq \| \bxi^{k-1} - \bxi^{\star} \|^2_2 - (2\beta_k - \beta_k^2 \alpha^2) \| \bar{\brho}^{\star}-\bar{\brho}^k \|^2_F.
\label{eq:final_form}
\end{equation}
Multiplying \eqref{eq:final_form} by $1/2$, and from the assumption that $\beta^k  \in (0, \alpha)$, we have that $\beta_k - \beta^2_k\alpha^2/2>0$.
Thus, the sequence $\{ \bxi^k \}$ is strictly non-increasing, and has limit point $\bxi^{\star}$.
This implies that $\| \brho^k - \brho^{\star} \|^2_F \rightarrow 0$ as $k \rightarrow \infty$, and proves convergence.

\section{Section \ref{sec:PA_analysis_gaurantees} Proofs}

\subsection{Proof of Lemma \ref{thm:suff_cond}}\label{sec:suff_cond_prf}

We start by lower bounding $|\theta(s,\bar{\bi x}_k, \bar{\bi x}_l)|$.
Without loss of generality, we let $k'=l'$ and $\x_k=0$. 
This implies $k\neq l$ and $\x_l \neq 0$ since $(k,k',l,l')\in I_2$, and
\begin{equation}
|\theta(s,\bar{\bi x}_k, \bar{\bi x}_l)| = \big\vert |\bgamma_1(s)|-|\x_l-\bgamma_1(s)| + \hat{\y}\cdot\x_l \big\vert.
\end{equation}
Using the reverse triangle inequality and the fact that $|\bgamma_1(s)|>|\x_l|$, the receiver terms cancel and we have
\begin{equation}
|\theta(s,\bar{\bi x}_k, \bar{\bi x}_l)| \geq \big\vert -|\x_l| + \hat{\y}\cdot \x_l \big\vert.
\label{eq:cond2_inter}
\end{equation}
Setting $\hat{\y}\cdot \x_l = |\x_l|\cos\varphi(\x_l,\y)$, \eqref{eq:cond2_inter} becomes
\begin{equation}
|\theta(s,\bar{\bi x}_k, \bar{\bi x}_l)| \geq |\x_l||\cos\varphi(\x_l,\y)-1|.
\label{eq:cond2_inter2}
\end{equation}

Since the ground topography is flat, $\cos(\varphi(\x_l,\y)) = \cos\alpha_{\y}\cos\varphi(\x_l)$ where $\varphi(\x_l)$ is the angle between $\x_l$ and the projection of $\y$ onto the ground plane.
Noting that $\cos\varphi(\x_l) \in [-1,1]$ and $\cos\alpha_{\y} \in (0,1)$ is fixed, \eqref{eq:cond2_inter2} is lower bounded when $\cos\alpha_{\y}\cos\varphi(\x_l)$ is largest. This occurs when $\cos\varphi(\x_l)=1$.
Thus,
\begin{equation}
|\theta(s,\bar{\bi x}_k, \bar{\bi x}_l)| \geq |\x_l||1-\cos\alpha_{\y}|. 
\label{eq:theta_lower_bnd_2}
\end{equation}

Taking the minimum of \eqref{eq:theta_lower_bnd_2} with respect to $(k,k',l,l')\in I_2$ we have
\begin{equation}
\min\limits_{(k,k',l,l')\in I_2} |\theta(s,\bar{\bi x}_k, \bar{\bi x}_l)| \geq \min\limits_{(k,k',l,l')\in I_2} |\x_l||1-\cos\alpha_{\y}|. 
\label{eq:theta_lower_bnd_3}
\end{equation}
Since the right hand side of \eqref{eq:theta_lower_bnd_3} does not depend on $(k',l')$ it reduces to
\begin{equation}
\min\limits_{(k,k',l,l')\in I_2} |\theta(s,\bar{\bi x}_k, \bar{\bi x}_l)| \geq \Delta_x|1-\cos\alpha_{\y}|,
\label{eq:theta_lower_bnd_4}
\end{equation}
since $\min\limits_{k,l} |\x_k-\x_l| = \min\limits_{k,l} |\x_l| = \Delta_x$.
Multiplying \eqref{eq:theta_lower_bnd_2} by $\omega_c/c_0$ we have
\begin{equation}
\min\limits_{(k,k',l,l')\in I_2} \frac{\omega_c |\theta(s,\bar{\bi x}_k, \bar{\bi x}_l)|}{c_0} \geq  \frac{\omega_c|\Delta_x||\cos\alpha_{\y}-1|}{c_0}.
\label{eq:res_nec_cond_2}
\end{equation}
Substituting \eqref{eq:resolution} into \eqref{eq:res_nec_cond_2} gives \eqref{eq:nec_bnd_2}, which completes the proof.

\subsection{Proof of Lemma \ref{thm:int_est}}\label{sec:int_est_pf}


In the first case, $(k,k',l,l') \in I_1$, which implies $\bar{\bi x}_k = \bar{\bi x}_l$ and the phase of the exponent in \eqref{eq:inner_int_2} is equal to zero, thus \eqref{eq:inner_int_2} becomes 
\begin{equation}
\begin{aligned}
K(\bar{\bi x}_k, \bar{\bi x}_l) & =  \frac{1}{(4\pi)^4|\y|^4} \int_S B_1(s,\bar{\bi x}_k, \bar{\bi x}_k) ds \\ & \qquad \times \int_{\Omega} B_2(\omega,\bar{\bi x}_k, \bar{\bi x}_k) d\omega \\ & = \frac{C_J}{(4\pi)^4|\y|^4} \int_S B_1(s,\bar{\bi x}_k, \bar{\bi x}_k) ds
\end{aligned}
\label{eq:K_int_proof_1}
\end{equation}
where the second equality follows from \eqref{eq:cj_def}.
To approximate the $s$ integration in \eqref{eq:K_int_proof_1}, we first derive estimates of the integrand.
First, observe that
\begin{equation}
\begin{aligned}
|\x_k-\bgamma_1(s)|^2 \approx |\bgamma_1(s)|^2 - 2\x_k\cdot\bgamma_1(s),
\end{aligned}
\end{equation}
since $|\bgamma_1(s)|^2 \gg |\x_k|^2$.
The same approximation applies to $1/|\x_k-\bgamma_2(s)|^2$.
\begin{equation}
B_1(s,\bar{\bi x}_k, \bar{\bi x}_k) \approx \frac{1}{|\bgamma_1(s)|^2|\bgamma_2(s)|^2}
\end{equation}
since $|\bgamma_i| \gg |\x_k|$.
Furthermore, we approximate 
\begin{equation}
B_1(s,\bar{\bi x}_k, \bar{\bi x}_k) \approx \frac{1}{L_{g_1}^2 L_{g_2}^2}
\end{equation}
where $L_{g_i}$ is the range of the $i$-th antenna at the center of the synthetic aperture\footnote{For typical flight trajectories this is the shortest range.}
Thus,
\begin{equation}
\int_S B_1(s,\bar{\bi x}_k, \bar{\bi x}_k) ds \approx \frac{\Delta_s}{L_{g_1}^2 L_{g_2}^2}.
\label{eq:B1_approx}
\end{equation}
Substituting \eqref{eq:B1_approx} into \eqref{eq:K_int_proof_1} we obtain
\begin{equation}
K(\bar{\bi x}_k, \bar{\bi x}_l) \approx \frac{C_J \Delta_s}{(4\pi)^4|\y|^4 L_{g_1}^2 L_{g_2}^2}.
\label{eq:K_est_I1}
\end{equation}

For the second case, $(k,k',l,l') \in I_2$, 
and isolating the $\omega$ integration in \eqref{eq:inner_int_2} 
the integral becomes
\begin{equation}
\begin{aligned}
K(\bar{\bi x}_k, \bar{\bi x}_l) & = \frac{1}{(4\pi)^4|\y|^4}\int_S e^{-i\frac{\omega_c}{c_0} \theta(s,\bar{\bi x}_k, \bar{\bi x}_l)} B_1(s,\bar{\bi x}_k, \bar{\bi x}_l) \\ & \qquad  \times \left[ \int\limits_{-B/2}^{B/2} e^{-i\frac{\omega}{c_0} \theta(s,\bar{\bi x}_k, \bar{\bi x}_l)} B_2(\omega,\bar{\bi x}_k, \bar{\bi x}_l) d\omega \right] ds.
\label{eq:K_int_proof_2}
\end{aligned}
\end{equation}
Then, substituting the approximation of $B_1(s,\bar{\bi x}_k, \bar{\bi x}_l)$ given in \eqref{eq:B1_approx} and \eqref{eq:W_def} into \eqref{eq:K_int_proof_2}, we have
\begin{equation}
\begin{aligned}
K(\bar{\bi x}_k, \bar{\bi x}_l) & \approx \frac{1}{(4\pi)^4|\y|^4 L^2_{g_1} L^2_{g_2}}\int_S e^{-i\frac{\omega_c}{c_0} \theta(s,\bar{\bi x}_k, \bar{\bi x}_l)} \\ & \qquad \qquad \qquad \qquad \ \ \times   W(s,\bar{\bi x}_k,\bar{\bi x}_l) ds.
\end{aligned}
\label{eq:K_int_proof_3}
\end{equation}

Since $W(s,\bar{\bi x}_k,\bar{\bi x}_l)$ satisfies \eqref{eq:mosp_assump}, and the condition of Lemma \ref{thm:suff_cond} is satisfied, we can evaluate the $s$ integration using the method of stationary phase, stated in Appendix \ref{sec:mosp}.
We begin by noting that the only values of $s$ that contribute to the integral are those that satisfy:
\begin{equation}
\begin{aligned}
\dot{\theta}(s_0,\bar{\bi x}_k, \bar{\bi x}_l) = & \widehat{(\x_k-\bgamma_i(s_0))}\cdot\dot{\bgamma}_i(s_0) \\ & - \widehat{(\x_{k'}-\bgamma_j(s_0))}\cdot\dot{\bgamma}_j(s_0) \\ & -  \widehat{(\x_l-\bgamma_i(s_0))}\cdot\dot{\bgamma}_i(s_0) \\ & + \widehat{(\x_{l'}-\bgamma_j(s_0))}\cdot\dot{\bgamma}_j(s_0) = 0.
\end{aligned}
\label{eq:theta_dot}
\end{equation}
Examining \eqref{eq:theta_dot}, we see that since $\bar{\bi x}_k \neq \bar{\bi x}_l$, the only stationary points are those for which $\widehat{(\x_k-\bgamma_i(s_0))}\cdot\dot{\bgamma}_i(s_0) - \widehat{(\x_{k'}-\bgamma_j(s_0))}\cdot\dot{\bgamma}_j(s_0) = 0$, of which there are only a finite number. 
Thus, without loss of generality, we consider one of these points, $s_0$, for which $\ddot{\theta}(s_0,\bar{\bi x}_k, \bar{\bi x}_l) \neq 0$.
Letting $\omega_c \rightarrow \infty$ we obtain the estimate:
\begin{equation}
K(\bar{\bi x}_k, \bar{\bi x}_l) \approx \frac{1}{(4\pi)^4|\y|^4 L^2_{g_1}L^2_{g_2}} \left\{ \left( \frac{2\pi}{\omega_c} \right)^{1/2} C_K + \Oc\left( \frac{1}{\omega_c^{3/2}} \right) \right\}
\end{equation}
where 
\begin{equation}
\begin{aligned}
C_K & = \sum\limits_{\{ s_0 \vert \dot{\theta}(s_0,\bar{\bi x}_k, \bar{\bi x}_l)=0 \}} \frac{e^{i\omega_c\theta(s_0,\bar{\bi x}_k, \bar{\bi x}_l)}e^{i\pi/4\ddot{\theta}(\bar{\bi x}_k, \bar{\bi x}_l,s_0)}}{\sqrt{|\ddot{\theta}(s_0,\bar{\bi x}_k, \bar{\bi x}_l)|}} \\ & \qquad \qquad \qquad \qquad \times W(s_0,\bar{\bi x}_k, \bar{\bi x}_l).
\end{aligned}
\label{eq:K_est_I2}
\end{equation}
Combining \eqref{eq:K_est_I1} and \eqref{eq:K_est_I2}, we have the final result stated in \eqref{eq:K_est}.

%

\subsection{Proof of Theorem \ref{thm:gaurantee}}\label{sec:gaurantee_pf}

We begin by determining the RIC for the forward model, given in \eqref{eq:forward_model_chap4}. This is done by first evaluating $\|\mathbf{F}\brho\|_2$. Starting with \eqref{eq:l2_norm_data}, we have
 \begin{equation}
\| \mathbf{F}\brho \|^2_2  = \sum\limits_{k,l=1}^{N^2} \sum\limits_{l,l'=1}^{N^2} K(\bar{\bi x}_k,\bar{\bi x}_l) \rho(\bar{\bi x}_k)\rho^*(\bar{\bi x}_l)
\label{eq:L2_norm}
\end{equation}
where the function $K$ is defined by the sum \eqref{eq:inner_sum}, and can be approximated as in \eqref{eq:inner_int_2}.
Separating the summations into two terms using the index sets $I_1$ and $I_2$, we have
\begin{equation}\label{eq:2_norm_split}
\begin{aligned}
\| \mathbf{F}\brho \|^2_2 = & \sum_{(k,k',l,l')\in I_1} K(\bar{\bi x}_k,\bar{\bi x}_l) \rho(\bar{\bi x}_k)\rho^*(\bar{\bi x}_l) \\ & \quad + \sum_{(k,k',l,l')\in I_2} K(\bar{\bi x}_k,\bar{\bi x}_l) \rho(\bar{\bi x}_k)\rho^*(\bar{\bi x}_l).
\end{aligned}
\end{equation}
Using the estimate of the integral $K$ for each set, given in Lemma \ref{thm:int_est}, we obtain
\begin{equation}\label{eq:2_norm_split_2}
\begin{aligned}
\| \mathbf{F}\brho \|^2_2 \approx & \frac{ C_J \Delta_s}{(4\pi)^4|\y|^4 L_{g_1}^2 L_{g_2}^2} \sum_{(k,k',l,l')\in I_1} \rho(\bar{\bi x}_k)\rho^*(\bar{\bi x}_l) \\ & \quad + \frac{1}{(4\pi)^4|\y|^4 L_{g_1}^2 L_{g_2}^2} \left\{ \left( \frac{2\pi}{\omega_c} \right)^{1/2} C_K \right. \\ & \left. + \Oc\left( \frac{1}{\omega_c^{3/2}} \right) \right\} \sum_{(k,k',l,l')\in I_2} \rho(\bar{\bi x}_k)\rho^*(\bar{\bi x}_l) \\
\approx & \frac{ C_J \Delta_s}{(4\pi)^4|\y|^4 L_{g_1}^2 L_{g_2}^2} \left[ \|\brho\|^2_F + \left\{ \Oc\left( \frac{1}{\omega_c^{1/2}} \right) \right. \right. \\ & \left. \left. + \Oc\left( \frac{1}{\omega_c^{3/2}} \right) \right\} \sum_{(k,k',l,l')\in I_2} \rho(\bar{\bi x}_k)\rho^*(\bar{\bi x}_l) \right].
\end{aligned}
\end{equation}
Clearly, 
as $\omega_c \rightarrow \infty$, 
\begin{equation}
\| \mathbf{F}\brho \|^2_2 \ \sim \frac{\Delta_s C_J}{(4\pi)^4|\y|^4 L_{g_1}^2 L_{g_2}^2} \| \brho \|^2_{F}.
\end{equation}
Therefore, up to normalization, 
we have
\begin{equation}
\| \mathbf{F}\brho \|_2 \ \sim \| \brho \|_{F}.
\label{eq:F_norm_final}
\end{equation}
From \eqref{eq:F_norm_final} we see that $\mathbf{F}$ is approximately an isometric mapping, which implies that the RIP constant is zero.
Therefore, by Theorem \ref{thm:recht_thm2} the solution $\mathrm{P}_4$ is the unique rank-one solution. 
Hence, the locations and reflectivities of the scatterers can be recovered exactly as the eigenvalue/vector pair of the recovered Kronecker image. 


\subsection{Proof of Theorem \ref{thm:phaseless_thm}}\label{sec:phaseless_gaurantee_pf}

The proof is very similar to that of Theorem \ref{thm:gaurantee}, given in Appendix \ref{sec:gaurantee_pf}.
Therefore, we omit details which are repetitive and note the appropriate changes that need to be made.
As before, we estimate
\begin{equation}
\| \mathbf{F}\brho \|^2_2 = \sum\limits_{(k,k',l,l')\in \{ \tilde{I}_1 ,\tilde{I}_2 \}} K(\bar{\bi x}_k,\bar{\bi x}_l) \rho(\bar{\bi x}_k)\rho^*(\bar{\bi x}_l) 
\end{equation}
by splitting the summation over the index sets $\tilde{I}_1, \tilde{I}_2$ and estimating each separately, as follows
\begin{equation}\label{eq:phaseless_2_norm_split_1}
\begin{aligned}
\| \mathbf{F}\brho \|^2_2 = & \sum_{(k,k',l,l')\in \tilde{I}_1} K(\bar{\bi x}_k,\bar{\bi x}_l) \rho(\bar{\bi x}_k)\rho^*(\bar{\bi x}_l) \\ 
& \quad + \sum_{(k,k',l,l')\in \tilde{I}_2} K(\bar{\bi x}_k,\bar{\bi x}_l) \rho(\bar{\bi x}_k)\rho^*(\bar{\bi x}_l).
\end{aligned}
\end{equation}
Using the estimate of the integral $K$, the fact that $\tilde{I}_1 = I_1 \cup I_3$, and letting $\omega_c \rightarrow \infty$ \eqref{eq:phaseless_2_norm_split_1} becomes
\begin{equation}\label{eq:phaseless_2_norm_split_2}
\begin{aligned}
\| \mathbf{F}\brho \|^2_2 \sim & \sum_{(k,k',l,l')\in I_1} \rho(\bar{\bi x}_k)\rho^*(\bar{\bi x}_l) \\ 
& \quad + \sum_{(k,k',l,l')\in I_3} \rho(\bar{\bi x}_k)\rho^*(\bar{\bi x}_l).
\end{aligned}
\end{equation}
In \eqref{eq:phaseless_2_norm_split_2}, the first term is the square of the Frobenius norm, and the second term is the square of the trace.
Thus,
\begin{equation}\label{eq:phaseless_2_norm_split_3}
\| \mathbf{F}\brho \|^2_2 \sim \| \brho \|^2_{F} + \left( \text{Tr}(\brho) \right)^2.
\end{equation}

Using $0 \lesssim \text{Tr}(\brho) \leq \sqrt{r}\| \brho \|_F$ in \eqref{eq:phaseless_2_norm_split_3} we have
\begin{equation}
\label{eq:phaseless_norm_bnd_1}
\| \brho \|_F^2 \lesssim \| \mathbf{F}\brho \|^2_2 \lesssim (1 + r)\| \brho \|_F^2,
\end{equation}
where $r = \text{rank}(\brho)$,
and taking the square root, \eqref{eq:phaseless_norm_bnd_1} becomes
\begin{equation}
\| \brho \|_F \lesssim \| \mathbf{F}\brho \|_2 \lesssim \sqrt{1 + r}\| \brho \|_F.
\end{equation}
Using the definition of RIP and equating the lower and upper bounds to $1-\delta_r$ and $1+\delta_r$, respectively, we have that the RIC for $\delta_{2r} = \max(0,1 - \sqrt{3}) < 1$.
Thus, from Theorem \ref{thm:chai_thm} the solution of $\mathrm{P}_3$ and $\mathrm{P}_1$ is equivalent and unique, and the exact locations and intensities of the scatterers is recovered.

\end{appendices}

\bibliographystyle{IEEEtran}
\bibliography{IEEEabrv,ref,biblioI}

\end{document}